\documentclass[journal]{IEEEtran}
\usepackage{amsthm}
\usepackage{cite}
\usepackage{amsmath,amssymb,amsfonts,bbm}
\usepackage{graphicx}
\usepackage{subfigure}
\usepackage[linesnumbered,ruled]{algorithm2e}
\usepackage{color}

\newtheorem{lemma}{\textbf{\textit{Lemma}}}

\newtheorem{theorem}{\textbf{\textit{Theorem}}}

\allowdisplaybreaks[3] % 允许长公式跨页
\usepackage{booktabs}
\usepackage{threeparttable}

\usepackage[colorlinks,linkcolor=blue]{hyperref}
\usepackage{stfloats}
\normalsize

\ifCLASSINFOpdf

\else

\fi



\begin{document}
	\title{Joint Port Selection Based Channel Acquisition \\for FDD Cell-Free Massive MIMO}
	
\author{ Cheng Zhang,~\IEEEmembership{Member,~IEEE,} Pengguang Du, Minjie Ding,\\  Yindi Jing,~\IEEEmembership{Senior Member,~IEEE,} Yongming Huang,~\IEEEmembership{Senior Member,~IEEE}
 % <-this % stops a space
 \thanks{This work was supported in part by the National Natural Science Foundation of China under Grant 62271140 and 62225107, the Fundamental Research Funds for the Central Universities 2242022k60002, the Natural Science Foundation on Frontier Leading Technology Basic Research Project of Jiangsu under Grant BK20222001, and the Major Key Project of PCL. (Corresponding authors: Y. Huang, C. Zhang)}
\thanks{C. Zhang, P. Du, M. Ding and Y. Huang are with the National Mobile Communication Research Laboratory, the School of Information Science and Engineering, Southeast University, Nanjing 210096, China, and also with the Purple Mountain Laboratories, Nanjing 211111, China (e-mail: $ \left\lbrace \right.  $zhangcheng\_seu; pgdu; 220200816; huangym$ \left.  \right\rbrace $@seu.edu.cn).}% <-this % stops a space
\thanks{Y. Jing is with the Department of Electrical and Computer Engineering, University of Alberta, Edmonton T6G 1H9, Canada (e-mail: yindi@ualberta.ca).}}

	\markboth{IEEE Transactions on COMMUNICATIONS}%
	{Submitted paper}
	\maketitle
	%\vspace{-4em}
	\begin{abstract}
		%\vspace{-1em}
	In frequency division duplexing (FDD) cell-free massive MIMO, the acquisition of the channel state information (CSI) is very challenging because of the large overhead required for the training and feedback of the downlink channels of multiple cooperating base stations (BSs). 
	In this paper, for systems with partial uplink-downlink channel  reciprocity, and a general spatial domain channel model with variations in the average port power and correlation among port coefficients, we propose a joint-port-selection-based CSI acquisition and feedback 
	scheme for the downlink transmission with zero-forcing precoding.
	The scheme uses an eigenvalue-decomposition-based transformation to reduce the feedback overhead by exploring the port correlation.
	We derive the sum-rate of the system for any port selection.
	Based on the sum-rate result, we propose a low-complexity greedy-search-based joint port selection (GS-JPS) algorithm. Moreover, to adapt to fast time-varying scenarios, a supervised deep learning-enhanced joint port selection (DL-JPS) algorithm is proposed. Simulations verify the effectiveness of our proposed schemes and their advantage over existing port-selection channel acquisition schemes.
	\end{abstract}

	\vspace{0em}
	\begin{IEEEkeywords}
	%	\vspace{-1em}
		Cell-free massive MIMO, channel acquisition, joint port selection, channel correlation, deep learning.
		
	\end{IEEEkeywords}

	% \fi
	\IEEEpeerreviewmaketitle

	\section{Introduction}
	%\textcolor{blue}{**text in blue color was taken from another paper and needs to be rewritten**}
	%\vspace{-0.2cm}
	
	\IEEEPARstart{M}{assive} MIMO technology can significantly improve the capacity and reliability of wireless networks. Via configuring arrays with hundreds of antennas at the base stations (BSs), it is possible to serve dozens of single-antenna users in the same time-frequency resource, thus meeting the high multiplexing requirements in multi-user scenarios \cite{lu2014overview,larsson2014massive}.
	At the same time, network densification has received great attention in 5G wireless communications \cite{wp5d2015imt}.  
	In the existing cellular network architecture, however, the reduction of cell size leads to a sharp increase in the number of users at the cell edge, making it difficult to achieve the expected network throughput. 
	To address the problem, the concept of cell-free has been proposed in \cite{ngo2017cell}, where a group of BSs collaborate to serve the users, by which a significant reduction in multi-user interference can be achieved and the so-called cell edge effect is greatly diminished\cite{wang2023full}.

	In order to maximize the BS cooperation gain for the cell-free downlink, the quality of the downlink channel state information (CSI) at the BS is crucial\cite{jin2019channel}. 
	For the time division duplexing (TDD) massive MIMO systems, it is relatively easy for the BS to acquire downlink CSI due to the reciprocity between the uplink and downlink channels. 
	However, for the frequency division duplexing (FDD) massive MIMO systems, to achieve the rate comparable to the systems with perfect CSI \cite{jindal2006mimo}, the number of pilots and especially the CSI feedback overhead are both proportional to the number of transmit antennas at the BS side \cite{lee2015antenna}, \cite{shen2015joint}.
	For this characteristic, most of the work on massive MIMO has focused on TDD. However, 
	FDD is still widely adopted in current 5G communication systems with more mature industrial products and market share, and many frequency bands are allocated for the FDD systems \cite{liang2016fdd}.
	In addition, since the TDD systems have the operation of constantly switching between uplink and downlink transmission, uplink access of users as well as instantaneous channel estimation at the BSs are challenging, whereas the FDD systems have the advantages of low transmission delay, continuous channel estimation, and broad range of coverage \cite{lee2015antenna}.
	Considering the realistic scenario, the TDD systems requires the reciprocity calibration of the transmitter and receiver radio frequency (RF) chains. Non-ideal calibration also creates the need for downlink channel estimation, and the feedback overhead cannot be avoided \cite{jiang2015achievable}.
	The CSI acquisition task is even more challenging in FDD cell-free massive MIMO systems since the users need to estimate and feed back the downlink CSI of multiple BSs. 
	%\textcolor{blue}{This motivates us to investigate solutions to the potential problems in FDD cell-free systems.}

	Current studies mainly focus on the domain knowledge model to fully exploit the partial reciprocity of the uplink and downlink channels in the FDD system to reduce the feedback overhead\cite{yin2022partial,liu2022learning}.
	In typical environments, the relative permittivity and conductivity of the obstacle do not change significantly in the tens of gigahertz range, i.e., the reflection and deflection characteristics are almost identical. Therefore, the angle of departure (AoD) of the downlink channel is almost the same as the angle of arrival (AoA) of the uplink channel \cite{zhong2020fdd}.
	In \cite{abdallah2020efficient}, based on the discrete Fourier transform (DFT) operation and the log-likelihood function, a method for the multipath component estimation was proposed, which provides a significant enhancement over traditional approaches in terms of mean-squared-error (MSE) of the estimated AoA and large-scale fading coefficients.
	Compared to these statistical schemes, it is interesting to explore instantaneous CSI to realize higher cell-free performance gain.
	
	To reduce the feedback overhead of FDD cell-free massive MIMO systems, a partial reciprocity-based port-selection feedback framework was proposed in \cite{kim2020downlink} where the BSs obtain the multipath AoD of the downlink channel through uplink pilots, and select some dominant paths for the sum-rate maximization. The BSs send downlink precoded pilots from which the users obtain the gains of the selected paths and feed them back to the BSs. 
	The path selection problem is solved in an alternating way by determining a signal-to-leakage-and-noise-ratio (SLNR)-type precoding for given path selection and updating the path selection via removing the path index with the minimal impact on the SLNR.
	%\textcolor{blue}{The impact of differences in the large-scale fading of paths is missing in this work.}
	Large-scale-fading differences between different BS-User links and different paths on the same link are not considered.
	In \cite{kim2021energy}, a joint dominate path selection and power allocation algorithm was proposed for the energy efficiency maximization in millimeter wave cell-free systems, where the weighted sum of array responses for the selected paths is adopted as the precoder while the path selection is conducted in a non-heuristic manner via the formulation of the sparse support identification problem.

	% 首先SINR指标相比SLNR指标更为直接，面向SINR最大化的线性ZF也是更为实际系统青睐的预编码方式
	% 做减法式的端口选择方式在备选端口数量很大，而可选端口数量较少时，开销较大
	% 已有研究在推到端口选择的性能指标时，没有考虑实际中可能的端口系数功率不一致性问题
	%\textcolor{blue}{Some deficiencies exist in the current work. There is some performance loss in \cite{zhong2020fdd} which focuses on multipath angle and large-scale fading estimation and does not consider small-scale channel feedback. Whereas in \cite{kim2020downlink}, the SLNR precoding is used and the effect of large-scale fading is not considered, and it is mainly optimized for energy efficiency in \cite{kim2021energy}.}
	Compared to the SLNR precoding, the signal-to-interference-plus-noise-ratio (SINR) maximization-orientated zero-forcing (ZF) precoding is more relevant to the sum-rate performance and preferred in the high signal-to-noise-ratio (SNR) or interference-limited scenarios\cite{wiesel2008zero}. 
	The subtractive port selection in \cite{kim2020downlink} results in a high overhead when the number of ports is large and the number of selected ports is relatively small.
	In addition, the difference among the average power of different paths \cite{you2015pilot} is not taken into account in the derivation of the path selection metric in the existing works.

In addition to the uplink and downlink partial reciprocity in the FDD system, there are other channel characteristics for cell-free systems that can be explored for the downlink CSI acquisition.
In \cite{cheng2012cooperative}, \cite{zhou2019geometry} and \cite{zhang2020novel}, the channel correlation of multiple collaborative BSs was investigated via geometric statistical channel modeling, with emphasis on the effect of the local scatterer density at the user side on this correlation. In \cite{zhou2019geometry}, the channel correlation of adjacent collaborative BSs in a high-speed-railway wireless communication scenario was verified based on real measurements. 
The channel correlation between adjacent users served by the same BS has been exploited to reduce the CSI feedback overhead for massive MIMO systems in \cite{guo2020deep}. The potential correlation between the channels of multiple collaborating BSs in cell-free MIMO systems has yet to be fully utilized.

	% 信道相关性问题------
	% 除了FDD本身的上下行准互易性，有没有其他针对cell-free系统的信道特性可以挖掘呢？文献 [21]、 [22] 与 [23] 通过基于几何统计信道建模的方式研究了多协作基站信道的相关性，并着重分析了用户侧本地散射体密度对该相关性的影响。其中文献 [22] 还对高铁无线通信场景下相邻协作基站信道的相关性进行了实测验证。 借鉴传统大规模 MIMO 系统信道反馈方案设计中通过挖掘同一基站的天线间信道相关性与被同一基站服务的邻近用户间信道相关性从而降低反馈开销的思路，无蜂窝大规模 MIMO 系统中多个协作基站信道间潜在的相关性也需要被利用。
	
	% 研究内容介绍

In this paper, we study the joint-port-selection-based channel acquisition for FDD cell-free massive MIMO downlink. Different from existing works, we consider the ZF precoding, a more realistic channel spatial power profile where the average power of the channel varies in different port directions and the coefficient correlation between ports of the same BS and between ports of different BSs. These distinctions make the port selection and port coefficient feedback designs completely different. Our major contributions are summarized as follows.
\begin{itemize}
\item We propose a joint-port-selection-based CSI feedback and reconstruction scheme for FDD cell-free massive MIMO downlink with ZF precoding. 
%\textcolor{red}{The scheme jointly explores the partial reciprocity of uplink and downlink channels and the coefficient correlation between different channel ports of multiple collaborative BSs.}
The scheme uses an eigenvalue-decomposition-based transformation (EDT) approach to reduce the overhead of the port coefficient feedback by exploiting the correlation between the multi-BS port coefficients.
\item An expression is derived for the sum-rate of the FDD cell-free massive MIMO downlink with ZF precoding and port selection. The result shows explicitly the effect of the different average channel strengths for different ports and the possible correlation among the port coefficients, and enables the sum-rate maximization-orientated joint port selection.
\item We propose a low-complexity greedy-search-based joint port selection (GS-JPS) algorithm based on the derived sum-rate expression, which improves the search efficiency by rationally setting the update priority of users, BSs, and ports. In addition, a supervised deep learning (DL)-enhanced joint port selection (DL-JPS) algorithm is proposed, which can adapt to fast time-varying scenarios.
\item Simulations validate our derived closed-form sum-rate expression and demonstrate that the GS-JPS algorithm and the EDT port coefficient feedback approach work together to achieve a higher sum-rate compared to existing port-selection-based channel acquisition schemes, especially in the medium-to-high SNR scenario. 
Moreover, the DL-JPS algorithm can quickly obtain a port selection with comparable performance to that of the GS-JPS algorithm, showing its suitability for fast time-varying scenarios.
\end{itemize}

	% 工作贡献介绍
	% 贡献1:同时利用了FDD信道上下行部分互易性，以及多基站信道相关性
	% 贡献2:推导了根据一般性假设（路径大尺度不要求相同）、且基站间端口系数潜在的相关性体现
	% 贡献3:模型数据协同驱动，模型方法帮助解决端口选择指标设计问题与端口系数的低开销反馈问题，并提供数据驱动的端口选择良好初始解，数据驱动则进一步优化端口选择

	The rest of this paper is organized as follows. In Section \ref{System Model}, we introduce a typical cell-free massive MIMO scenario with correlated channels and describe the proposed joint-port-selection-based CSI acquisition and feedback scheme. 
	In Section \ref{Sum-Rate Analysis}, an analytical sum-rate expression based on ZF precoding for arbitrary port selection is derived.
	In the following Section \ref{Port Selection Design Schemes},
	we formulate the optimization problem
	for port selection and design two algorithms, the GS-JPS algorithm and the DL-JPS algorithm, respectively.
	Detailed simulations and discussions are provided in Section \ref{Simulation and Discussion}. Section \ref{Conclusion} summarizes this work.

\emph{Notation}: In this paper, bold upper case letters and bold lower case letters denote matrices and vectors, respectively. 
		$\mathbb{{R}}^{m \times n}$, $\mathbb{{C}}^{m \times n}$ and $\mathbb{{B}}^{m \times n}$ denote the sets of $m$-by-$n$ matrices with real-valued entries, complex-valued entries, and binary-valued entries, respectively.
		$\mathbf{I}_{n}$ denotes the $n$-by-$n$ identity matrix. The
		conjugate transpose, transpose, trace, and determinant of $\mathbf{A}$ are denoted
		by $\mathbf{A}^{\rm{H}}$, $\mathbf{A}^{\rm{T}}$, ${\rm tr}\left\{ \mathbf{A}\right\} $ and $|\mathbf{A}|$.
		Also, $\left\Vert \mathbf{a}\right\Vert $ and
		$\left\Vert \mathbf{A}\right\Vert _{\rm{F}} $ denote the Euclidean norm of $\mathbf{a}$ and
		the Frobenius norm of $\mathbf{A}$, respectively. The vector $\mathbf{a}_m$ is the $m$-th column of the matrix $\mathbf{A}$.
		$\rm{diag}(\mathbf{a})$ is the diagonal matrix whose diagonal entries are elements of vector $\mathbf{a}$.
		${\rm{blk}}[\cdot]$ indicates the block-diagonal operator.
		$\mathbb{E}\left\{ \cdot\right\} $ 
		is the expected value operator and $f_{X}\left(\cdot\right)$
		denotes the probability density function
		(PDF) of the random variable
		$X$. 
		%$\cong$ denotes the equality of the distribution. 
		%$\Gamma(\cdot)$ represents the Gamma function. 
		$\mathbb{{CN}}\left(\text{\ensuremath{\boldsymbol{\mu}}},\boldsymbol{\Sigma}\right)$ denotes the circularly symmetric complex Gaussian distribution with
		mean vector $\boldsymbol{\mu}$ and covariance matrix $\boldsymbol{\Sigma}$. 
		$\chi _2^2\left( a \right)$ denotes a chi-squared distribution with $2$ degrees of freedom and the noncentrality parameter $a$.
		Additionally, $\left|\mathbb{\mathbb{A}}\right|$ is the cardinality of the set $\mathbb{A}$. 
		$\mathbb{\mathbb{A}}-\mathbb{\mathbb{B}}$ denotes the set of elements in $\mathbb{\mathbb{A}}$ but not in $\mathbb{\mathbb{B}}$.

	\section{System Model and Proposed CSI Feedback and Reconstruction Scheme}\label{System Model}
	%\begin{figure}[htbp]
		%\centering
		%\setlength{\abovecaptionskip}{0.cm}
		%\includegraphics[scale=0.7]{System_model.eps}
		%\caption{System model.}\label{system model}
		%\end{figure}
	We consider a typical FDD cell-free massive MIMO scenario, where in the coverage area of interest, $B$ BSs each with $M$ antennas serve $U$ single-antenna users via the collaborative transmission provided by the central unit (CU) through fronthaul links. 
	
	\subsection{Channel Model}\label{Multiple-BS Channel Model}
	We consider the typical uniform array, e.g., uniform linear array (ULA) or uniform planar array (UPA), at the BSs and the massive MIMO configuration, i.e., $M\gg 1$.
	The beamspace channel representation becomes a natural choice \cite{brady2013beamspace}, \cite{zhang2018interleaved} where the spatial DFT matrix $\mathbf{F}\in \mathbb{C}^{M \times M}$ is adopted to relate the antenna space and the beam space. Specifically, the downlink channel vector from the BS $b\in \mathbb{B}=\{1,... ,B\}$ to the User $u\in \mathbb{U}=\{1,... ,U\}$, denoted as $\mathbf{h}_{b, u}\in \mathbb{C}^{M\times1}$, can be written as
	\begin{equation}\label{channelmodel}
		\mathbf{h}_{b, u} = \sqrt{M}\mathbf{F}\mathbf{B}_{b,u}\bar{\mathbf h}_{b,u},
	\end{equation}
where $\mathbf{B}_{b,u}={\rm diag}\left(\left [\sqrt{{\bar \beta}_{b,u,1}},...,\sqrt{{\bar \beta}_{b,u,M}}\right ] \right)$ with ${\bar \beta}_{b,u,m}$ being the average power of the channel in the beam/port $\mathbf{f}_m$, the $m$-th column of $\mathbf{F}$  and $\bar{\bf h}_{b,u}=\left[\bar{h}_{b,u,1},...\bar{h}_{b,u,M}\right]^{\text{T}}$ is the vector of port coefficients. 
%\textcolor{red}{$\bar{\mathbf h}_{b,u}\sim \mathbb{CN}(\mathbf{0},\mathbf{I}_M)$ contains the channel small-scale coefficients of all beams/ports, i.e., the port coefficients.}
	In existing works, the port coefficients are generally assumed to be independent of each other \cite{heath2016overview}. However, due to the existence of common scatterers covered by beams of different ports, as shown in Fig. \ref{figure1_1},
	\begin{figure}[htp!]
		%\setlength{\abovecaptionskip}{0.1cm}   %调整图片标题与图距离
		%\vspace{-1em}
		\centering 
		\includegraphics[width=1\linewidth]{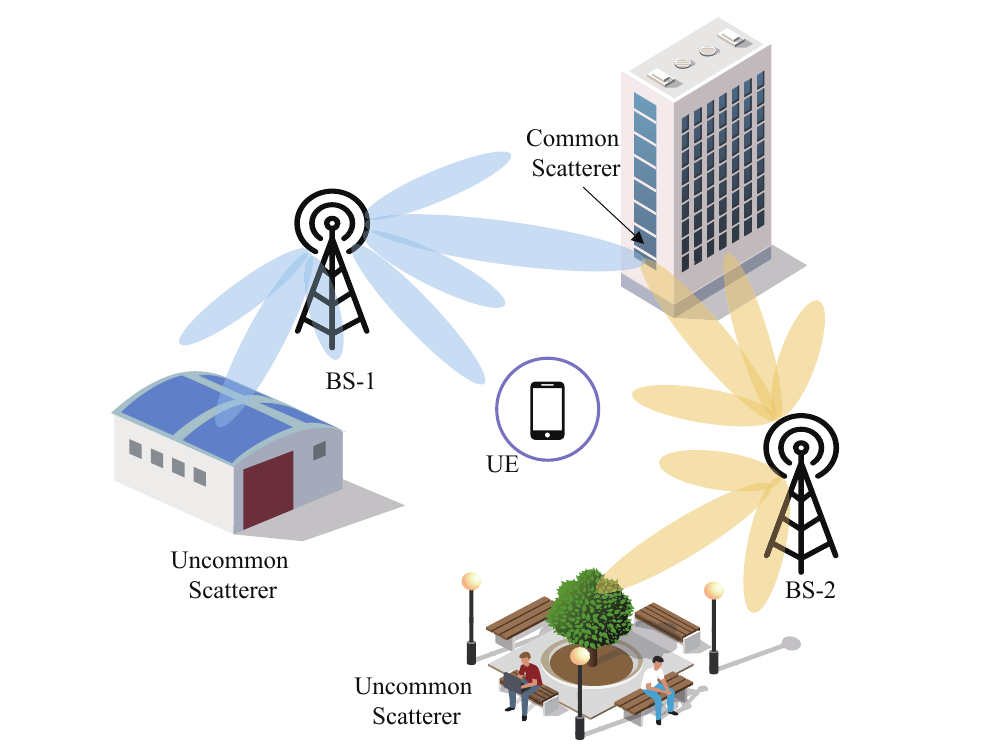}
		\caption{Schematic illustration of the existence of common scatterers covered by different beams.}
		\label{figure1_1}
		%\vspace{-2em}
	\end{figure}
	different port coefficients can be correlated.
	Therefore, the vector of the port coefficients of all BSs, denoted as  ${{{\bar{\mathbf h}}}_{u}}={{\left[ {{{{\bar{\mathbf h}}}}^{\rm{T}}_{1,u}},\ldots ,{{{{\bar{\mathbf h}}}}^{\rm{T}}_{B,u}} \right]}^{\rm{T}}}\in {{\mathbb{C}}^{BM\times 1}}$, follows
	$\mathbb{CN}\left(\mathbf{0},\mathbf{R}_{u} \right)$, where $\mathbf{R}_{u}$ is the covariance matrix of $\bar{\bf h}_u$. The $((b-1)M+l,(b'-1)M+l')$-th element of ${\bf R}_u$, denoted as $\rho_{u,b,b'}^{l,l'}$ or $\left[\mathbf{R}_{u}\right]_{(b-1)M+l,(b'-1)M+l'}$, is the correlation between the $l$-th port coefficient of BS $b$ and the $l'$-th port coefficient of BS $b'$ for User $u$. 
	%Further define $\check{\bf h}_u={\bf R}_u^{-1/2} \bar{\bf h}_u$. Thus $\check{\bf h}\sim \mathbb{CN}\left(\mathbf{0},\mathbf{I}_{BM}\right)$.
	Note that the coefficient correlation is constant in the statistical channel coherence time, which is typically much longer than the channel acquisition and feedback period.

	\subsection{Port-Selection-Based CSI Acquisition, Feedback, and Reconstruction}
	\label{Section_feedback_reconstruction}
	To reduce the overhead of CSI acquisition, one way is to select a small number of ports based on the downlink AoD and port statistics, and only estimate and feed back the channel coefficients of the selected ports.	
	The downlink channel statistics can be obtained with the help of the partial reciprocity between the uplink and downlink channels \cite{kim2020downlink}.
%	\textcolor{magenta}{
%		Partial reciprocity implies that the multi-path angles and the channel spatial power profile are equal for uplink and downlink. 
%		In practice, the downlink AoD and the average port power can be obtained with the help of the partial reciprocity between the uplink-downlink channels.
%		In addition, the downlink channel covariance matrix can be obtained by transforming the uplink covariance matrix acquired based on the uplink sounding reference signal to the downlink using the Hilbert space projection method in [ref].}
	
	Define $\Lambda_{b,u}=\{a_{b,u,1},...,a_{b,u,|\Lambda_{b,u}|}\}\subseteq\mathbb{M}=\{1,... ,M\}$ as the set of ports assigned to User $u$ by BS $b$ and $\Lambda^C_{b,u}=\mathbb{M}-\Lambda_{b,u}$. Define the set of selected ports and the set of remaining unselected ports of BS $b$ for users as ${{\Lambda }_{b}}=\underset{u\in \left\{ 1,\ldots, U \right\}}{\mathop{\bigcup }}\,{{\Lambda }_{b,u}}$ and $\Lambda^C_{b}=\mathbb{M}-\Lambda_{b}$, respectively.
	%Also, denote ${{\Lambda }_{u}}=\{{{\Lambda }_{1,u}},\ldots,{{\Lambda }_{B,u}} \}$.
	Define $\bar{\bf{h}}_{{\Lambda}_{b,u}}$ as the sub-vector of $\bar{\bf{h}}_{{b,u}}$ consisting of the elements of the set $\{{\bar h}_{b,u,m}, m\in{\Lambda_{b,u}}\}$. Equivalently, $\bar{\bf{h}}_{{\Lambda}_{b,u}} = {\bf{A}}_{b,u} \bar{\bf{h}}_{{b,u}}$ where the port selection matrix ${\bf{A}}_{b,u} \in \mathbb{B}^{|\Lambda_{b,u}| \times M}$ satisfies
	\begin{equation}
		\linespread{1.2} \selectfont
		[{\bf{A}}_{b,u}]_{i,j} = 	\begin{cases}
			1, &\text{if } j = a_{b,u,i}, \forall i=1,...,|\Lambda_{b,u}|  \\
			0, &\text{otherwise} \\
		\end{cases}.
	\end{equation}
	For example, for a system with $4$ antennas at each BS, if ports $2$ and $4$ of BS $b$ are selected for User $u$, we have
	\begin{equation}
		\linespread{1.2} \selectfont
		{\bf{A}}_{b,u} = \begin{bmatrix}
			0&1&0  &0 \\
			0& 0 &0  &1
		\end{bmatrix}, \text{and}\,\,
		\bar{\bf{h}}_{{\Lambda}_{b,u}} = 
		\begin{bmatrix}
			{\bar h}_{b,u,2}\\
			{\bar h}_{b,u,4}
		\end{bmatrix}.
	\end{equation}
	Define ${\bar{\mathbf{{h}}}_{\Lambda_u}}={{\left[ \bar{\mathbf{{h}}}_{\Lambda_{1,u}}^{\mathrm{T}},\ldots ,\bar{\mathbf{{h}}}_{\Lambda_{B,u}}^{\mathrm{T}} \right]}^{{\mathrm{T}}}}\in {{\mathbb{C}}^{K_u\times 1}}$ where $K_u=\sum\nolimits_{b=1}^{B}{\left| {{\Lambda }_{b,u}} \right|}$ is the total number of ports assigned to User $u$. Considering the possible correlation between port coefficients due to the common scatterers, we utilize the second-order covariance matrix of $ \bar{\bf{h}}_{{\Lambda}_{u}}$, i.e., 
	$\mathbf{R}_{{\Lambda}_{u}}=\mathbb{E}\left\{  \bar{\bf{h}}_{{\Lambda}_{u}} \bar{\bf{h}}_{{\Lambda}_{u}}^{{\mathrm{H}}} \right\} = {\bf{A}}_{u}{\bf{R}}_{u}{\bf{A}}^{\rm{H}}_{u}$
	as a priori knowledge, where ${\bf{A}}_{u} = {\rm{blk}}[{\bf{A}}_{1,u},\ldots,{\bf{A}}_{B,u}]$.
	%In the following, we first give the design for reducing the feedback overhead of the estimated port coefficients.
	In the following, we design an EDT approach that exploits the correlation among port coefficients to reduce their feedback overhead.
	
	 Under a general port selection $\Lambda_{b,u}$'s,
the BSs first send downlink pilots from the selected ports and the users perform the estimation of the downlink port coefficients to obtain ${\dot{\bar h}}_{b,u,m}, m\in{\Lambda_{b,u}}$. 
Denote ${\dot{\bar{\bf{h}}}}_{{\Lambda}_{b,u}}$ as the estimated vector of ${{\bar{\bf{h}}}}_{{\Lambda}_{b,u}}$.
For User $u \in \mathbb{U}$, by combining the estimated port coefficients ${\dot{\bar{\bf{h}}}}_{{\Lambda}_{b,u}}$ of all BSs into a vector, we obtain ${\dot{\bar{\bf{h}}}}_{{\Lambda}_{u}}= {{\left[ \dot{\bar{\mathbf{{h}}}}_{\Lambda_{1,u}}^{\mathrm{T}},\ldots ,\dot{\bar{\mathbf{{h}}}}_{\Lambda_{B,u}}^{\mathrm{T}} \right]}^{{\mathrm{T}}}}\in {{\mathbb{C}}^{K_u\times 1}}$.
Denote by ${{\breve{\bar{\bf{h}}}}}_{{\Lambda}_{u}} = {\bar{\mathbf{{h}}}_{\Lambda_u}} - {\dot{{\bar{\bf{h}}}}}_{{\Lambda}_{u}}$ the channel estimation error of User $u$. From
the minimum-mean-squared-error (MMSE) estimation property, ${{\breve{\bar{\bf{h}}}}}_{{\Lambda}_{u}}$ is independent of ${{\dot{\bar{\bf{h}}}}}_{{\Lambda}_{u}}$, and ${{\breve{\bar{\bf{h}}}}}_{{\Lambda}_{u}} \sim \mathbb{CN} (\mathbf{0}, \varepsilon_{\rm{CE}} ^2{ {{\bf{R}}_{\Lambda_u}}})$. Note that $\varepsilon_{\rm{CE}} ^2$ denotes the CSI estimation error variance, which is dependent on both the pilot energy and the port average power \cite{zhang2017performance,ngo2017total}.
 Denote the rank of $\mathbf{R}_{{\Lambda}_{u}}$ as $r_u$,
	and we consider the compact eigenvalue decomposition and eigenvalue decomposition
	\begin{equation}\label{equal3_16}
		%{\bar{\mathbf{{h}}}_{\Lambda_u}}\cong {{\mathbf{U}}_{u,\text{r}}}\boldsymbol{\Sigma} _{u,\text{r}}^{1/2}{{\tilde{\mathbf{r}}}_{u,\text{r}}},
		\mathbf{R}_{{\Lambda}_{u}}={{\mathbf{U}}_{u,\text{r}}}\boldsymbol{\Sigma} _{u,\text{r}}{{\mathbf{U}}^{\mathrm{H}}_{u,\text{r}}}
		\,
			 ={{\bf{U}}_{u}}{\rm{diag}}\left\{\boldsymbol{\Sigma} _{u,\text{r}},{\bf{0}}\right\} {{\bf{U}}^{\rm{H}}_{u}}.
	\end{equation}
where ${\boldsymbol{\Sigma}_{u,\text{r}}}\in {{\mathbb{C}}^{r_u\times r_u}}$ contains the $r_u$ non-zero eigenvalues of $\mathbf{R}_{{\Lambda}_{u}}$ as its diagonal elements, ${{\mathbf{U}}_{u,\text{r}}}\in {{\mathbb{C}}^{K_u\times r_u}}$ is composed of eigenvectors with non-zero eigenvalues,
${{\bf{U}}^{c}_{u,{\rm{r}}}}$ is the unitary complement of ${{\bf{U}}_{u,{\rm{r}}}}$, that is, ${{\bf{U}}_{u}} \triangleq \left[{{\bf{U}}_{u,{\rm{r}}}},{{\bf{U}}^{c}_{u,{\rm{r}}}}\right]$ is a $K_u \times K_u$ unitary matrix.

	For the channel feedback, User $u$ first calculates the dimension-reduced transformation-domain port coefficient vector ${{\mathbf{{r}}}_{u,\text{r}}}\in {{\mathbb{C}}^{r_u\times 1}}$ as 
	\begin{equation}\label{equal3_18}
		\setlength{\abovedisplayskip}{2pt}%
		\setlength{\belowdisplayskip}{2pt}%
		\setlength{\abovedisplayshortskip}{0pt}%
		\setlength{\belowdisplayshortskip}{0pt}% 
		{{\mathbf{{r}}}_{u,\text{r}}}=\boldsymbol{\Sigma} _{u,\text{r}}^{-1/2}\mathbf{U}_{u,\text{r}}^{\rm{H}}{\dot{\bar{\mathbf{{h}}}}_{\Lambda_u}},
	\end{equation}
	then quantizes ${{\mathbf{{r}}}_{u,\text{r}}}$ into ${{\mathbf{\overline{{r}}}}_{u,\text{r}}} ={{\mathbf{{r}}}_{u,\text{r}}}-{\tilde{\mathbf{{r}}}}_{u,\text{r}}$ by standard scalar quantization or vector quantization methods.
	Since ${{\mathbf{{r}}}_{u,\text{r}}} \sim \mathbb{CN}\left(\mathbf{0}, \left(1-\varepsilon_{\rm{CE}} ^2\right){ {{\bf{I}}_{r_u}}}\right)$, 
	 the quantization error vector can be approximated as ${\tilde{\mathbf{{r}}}}_{u,\text{r}} \sim (\mathbf{0}, \varepsilon_{\rm{Q}} ^2{ {{\bf{I}}_{r_u}}})$, where $\varepsilon_{\rm{Q}} ^2$ is the quantization error variance. ${\tilde{\mathbf{{r}}}}_{u,\text{r}}$ and ${{\mathbf{{r}}}_{u,\text{r}}}$ are independent \cite{wagner2012large}.
	 Finally, user feeds back the quantized values to the BS side.

For the CSI reconstruction, the BS side recovers the vector of the selected port coefficients  of all BSs for User $u$ 
${\hat{\bar{\mathbf{h}}}_{\Lambda_u}}$ as
\begin{equation}
	\setlength{\abovedisplayskip}{2pt}%
	\setlength{\belowdisplayskip}{2pt}%
	\setlength{\abovedisplayshortskip}{0pt}%
	\setlength{\belowdisplayshortskip}{0pt}% 
	{\hat{\bar{\mathbf{h}}}_{\Lambda_u}}={{\mathbf{U}}_{u,\text{r}}}\boldsymbol{\Sigma }_{u,\text{r}}^{1/2}{{\mathbf{\overline{{r}}}}_{u,\text{r}}}.
\end{equation}
It can be further expressed as
\begin{equation} \label{eq_7}
	\begin{aligned}
		{\hat{\bar{\mathbf{h}}}_{\Lambda_u}}&={{\bf{U}}_{u,{\rm{r}}}}{\bf{U}}_{u,{\rm{r}}}^{\rm{H}}{\bar {\bf{h}} _{{\Lambda _u}}} - {{\bf{U}}_{u,{\rm{r}}}}{\bf{U}}_{u,{\rm{r}}}^{\rm{H}}{{\breve{\bar{\bf{h}}}}}_{{\Lambda}_{u}} - {{\mathbf{U}}_{u,\text{r}}}\boldsymbol{\Sigma }_{u,\text{r}}^{1/2}{\tilde{\mathbf{{r}}}}_{u,\text{r}},\\
		&\overset{(a)}{=}  {{\bar{\mathbf{h}}}_{\Lambda_u}} - {{\breve{\bar{\bf{h}}}}}_{{\Lambda}_{u}} - {{\mathbf{U}}_{u,\text{r}}}\boldsymbol{\Sigma }_{u,\text{r}}^{1/2}{\tilde{\mathbf{{r}}}}_{u,\text{r}}.
	\end{aligned}
\end{equation}
Since  ${{{\bar{\bf{h}}}}}_{{\Lambda}_{u}}$ follows $ \mathbb{CN} (\mathbf{0}, { {{\bf{R}}_{\Lambda_u}}})$, we have
\begin{equation} \label{R-4}
	\setlength{\abovedisplayskip}{2pt}%
	\setlength{\belowdisplayskip}{2pt}%
	\setlength{\abovedisplayshortskip}{2pt}%
	\setlength{\belowdisplayshortskip}{3pt}% 
	\begin{aligned}
		{{\bf{U}}^{\rm{H}}_{u}}{\bar {\bf{h}} _{{\Lambda _u}}} = 
		\begin{pmatrix}
			{{\bf{U}}^{\rm{H}}_{u,{\rm{r}}}}{\bar {\bf{h}} _{{\Lambda _u}}}\\
			\left({{\bf{U}}^{c}_{u,{\rm{r}}}}\right)^{\rm{H}}{\bar {\bf{h}} _{{\Lambda _u}}}
		\end{pmatrix}
		\sim \mathbb{CN} \left({\bf{0}},{\rm{diag}}\left\{\boldsymbol{\Sigma} _{u,\text{r}},{\bf{0}}\right\}\right),
	\end{aligned}
\end{equation}
which shows that $\left({{\bf{U}}^{c}_{u,{\rm{r}}}}\right)^{\rm{H}}{\bar {\bf{h}} _{{\Lambda _u}}}$ is a random vector whose mean vector
and covariance matrix are both zero. Therefore, $\left({{\bf{U}}^{c}_{u,{\rm{r}}}}\right)^{\rm{H}}{\bar {\bf{h}} _{{\Lambda _u}}} = {\bf{0}}$.
From this and since ${\bf U}_u$ is unitary, we have,
\begin{equation}
	\setlength{\abovedisplayskip}{2pt}%
	\setlength{\belowdisplayskip}{3pt}%
	\setlength{\abovedisplayshortskip}{1pt}%
	\setlength{\belowdisplayshortskip}{2pt}% 
	\begin{aligned}
		\left[	{{\bf{U}}_{u,{\rm{r}}}}{{\bf{U}}^{\rm{H}}_{u,{\rm{r}}}}+{{\bf{U}}^{c}_{u,{\rm{r}}}}\left({{\bf{U}}^{c}_{u,{\rm{r}}}}\right)^{\rm{H}}
		\right] {\bar {\bf{h}} _{{\Lambda _u}}}  &= {\bar {\bf{h}} _{{\Lambda _u}}}, \\
		\Rightarrow
		{{\bf{U}}_{u,{\rm{r}}}}{\bf{U}}_{u,{\rm{r}}}^{\rm{H}}{\bar {\bf{h}} _{{\Lambda _u}}} &= {\bar {\bf{h}} _{{\Lambda _u}}}.
	\end{aligned}
\end{equation}
With the same reasoning, we have ${{\bf{U}}_{u,{\rm{r}}}}{\bf{U}}_{u,{\rm{r}}}^{\rm{H}}{{\breve{\bar{\bf{h}}}}}_{{\Lambda}_{u}} = {{\breve{\bar{\bf{h}}}}}_{{\Lambda}_{u}}$.
These explain the equality $(a)$ in Eq. \eqref{eq_7}.

 Denote ${\tilde{\bar{\bf{h}}}}_{{\Lambda}_{u}} \triangleq {{\breve{\bar{\bf{h}}}}}_{{\Lambda}_{u}} + {{\mathbf{U}}_{u,\text{r}}}\boldsymbol{\Sigma }_{u,\text{r}}^{1/2}{\tilde{\mathbf{{r}}}}_{u,\text{r}}$ as the CSI error introduced by the port coefficient estimation and feedback quantization. 
Since  ${{\mathbf{U}}_{u,\text{r}}}\boldsymbol{\Sigma }_{u,\text{r}}^{1/2}{\tilde{\mathbf{{r}}}}_{u,\text{r}} \sim \mathbb{CN}({\bf{0}},\varepsilon_{\rm{Q}} ^2{ {{\bf{R}}_{\Lambda_u}}})$, it can be shown that ${\tilde{\bar{\bf{h}}}}_{{\Lambda}_{u}} \sim \mathbb{CN}({\bf{0}},{ {{\tilde{\bf{R}}}_{\Lambda_u}}})$, where ${ {{\tilde{\bf{R}}}_{\Lambda_u}}} = \varepsilon^2{ {{\bf{R}}_{\Lambda_u}}}$ and $\varepsilon^2 = \varepsilon_{\rm{CE}} ^2+ \varepsilon_{\rm{Q}} ^2 $ indicates the overall error level
\footnote{For simplicity, we assume the same level of CSI error at the selected ports, but it is straightforward to extend to the scenario where the error level is inhomogeneous.}.
The vector ${\hat{\bar{\mathbf{h}}}_{\Lambda_u}}$ can be decomposed as  $	{\hat{\bar{\mathbf{h}}}_{\Lambda_u}}= {{\left[ {\hat{\bar{\mathbf{h}}}^{\mathrm{T}}_{\Lambda_{1,u}}},\ldots ,{\hat{\bar{\mathbf{h}}}^{{\mathrm{T}}}_{\Lambda_{B,u}}} \right]}^{{\mathrm{T}}}}$ where ${\hat{\bar{\mathbf{h}}}_{\Lambda_{b,u}}}$ contains the selected port coefficients of BS $b$ for User $u$.
Similarly, ${\tilde{\bar{\mathbf{h}}}_{\Lambda_u}}$ can be decomposed as  $	{\tilde{\bar{\mathbf{h}}}_{\Lambda_u}}= {{\left[ {\tilde{\bar{\mathbf{h}}}^{\mathrm{T}}_{\Lambda_{1,u}}},\ldots ,{\tilde{\bar{\mathbf{h}}}^{{\mathrm{T}}}_{\Lambda_{B,u}}} \right]}^{{\mathrm{T}}}}$.
BS $b$ gets the recovered channel vector from BS $b$ to User $u$ as
\begin{equation}
		\setlength{\abovedisplayskip}{8pt}%
	\setlength{\belowdisplayskip}{8pt}%
	\setlength{\abovedisplayshortskip}{6pt}%
	\setlength{\belowdisplayshortskip}{6pt}% 
\hat{\mathbf{h}}_{b,u}= \sqrt{M}
\mathbf{F}_{\Lambda_{b,u}}\mathbf{B}_{\Lambda_{b,u}} \mathbf{A}^{\rm{H}}_{{b,u}}{\hat{\bar{\mathbf{h}}}_{\Lambda_{b,u}}} ,
\end{equation}
 where 
 ${\mathbf{B}}_{\Lambda_{b,u}} = {\mathbf{A}}_{{b,u}} {\mathbf{B}}_{{b,u}}$ denotes the matrix consisting of the rows of $\mathbf{B}_{{b,u}}$ with indices belonging to $\Lambda_{b,u}$, and $\mathbf{F}_{\Lambda_{b,u}}$ is the matrix formed by $\mathbf{f}_m$, $m\in{\Lambda_{b,u}}$.

	\subsection{ZF Precoding with Recovered CSI and Transmission Model}
	Under the joint downlink transmission commonly used in cell-free massive MIMO systems, the received signal of User $u$ can be represented as
	\begin{equation}
		y_{u}=\sum_{b=1}^{B} \mathbf{h}_{b, u}^{\mathrm{H}} \mathbf{w}_{b, u} s_{u}+\sum_{v\neq u}^{U} \sum_{b=1}^{B} \mathbf{h}_{b, u}^{\mathrm{H}} \mathbf{w}_{b, v} s_{v}+n_{u},
	\end{equation}
where $\mathbf{w}_{b, u}\in \mathbb{C}^{M\times1}$ denotes the precoding vector of BS $b$ for User $u$, $s_{u}\sim \mathbb{CN}\left(0, 1\right)$ is the random data symbol for User $u$, and $n_u\sim \mathbb{CN}\left(0, \sigma_{n}^{2}\right)$ denotes the receiver additive noise. The precoding vectors $\mathbf{w}_{b, u}$'s are designed based on the CSI available at the BSs, i.e., $\hat{\mathbf{h}}_{b,u}$'s. Let $\mathbf{w}_u = \left[\mathbf{w}_{1,u}^{\mathrm{T}},...,\mathbf{w}_{B,u}^{\mathrm{T}}\right]^{\mathrm{T}}$, $\mathbf{W} = \left[\mathbf{w}_1,...,\mathbf{w}_{U}\right]$, $\hat{\mathbf{h}}_u = \left[\hat{\mathbf{h}}_{1,u}^{\mathrm{T}},...,\hat{\mathbf{h}}_{B,u}^{\mathrm{T}}\right]^{\mathrm{T}}$, and $\hat{\mathbf{H}} = \left[\hat{\mathbf{h}}_1,...,\hat{\mathbf{h}}_{U}\right]$.
	The ZF precoding is adopted in this work where
	\begin{equation}\label{eqn_ZF}
		\setlength{\abovedisplayskip}{1pt}%
		\setlength{\belowdisplayskip}{1pt}%
		\setlength{\abovedisplayshortskip}{0pt}%
		\setlength{\belowdisplayshortskip}{0pt}% 
		\mathbf{W} = \hat{\mathbf{H}}\left(\hat{\mathbf{H}}^{\rm{H}}\hat{\mathbf{H}}\right)^{-1}\boldsymbol{\Sigma}.
	\end{equation}
	In the above ZF precoding design, $\boldsymbol{\Sigma}={\rm diag}\left(\left[\omega_1,... ,\omega_U\right]\right)$ is the power scaling matrix, where 
	$\omega_u = \frac{\sqrt{P_u}}{\sqrt{{\mathbb E}\{\left\Vert \bar{\mathbf{w}}_{u}\right\Vert ^{2}\}}}$, $\bar{\mathbf{W}}=\hat{\mathbf{H}}\left(\hat{\mathbf{H}}^{\rm{H}}\hat{\mathbf{H}}\right)^{-1}$ and $P_u$ is the transmit power allocated to User $u$. 
	Based on this precoding, the received signal of User $u$  can be rewritten as
	\begin{equation}
		\setlength{\abovedisplayskip}{1pt}%
		\setlength{\belowdisplayskip}{1pt}%
		\setlength{\abovedisplayshortskip}{0pt}%
		\setlength{\belowdisplayshortskip}{0pt}% 
		y_{u} = \omega_u s_{u} + \sum_{v=1}^{U}\sum_{b=1}^{B} \tilde{\mathbf{h}}_{b, u}^{\mathrm{H}} \mathbf{w}_{b, v} s_{v} + n_{u},
	\end{equation}
	where $\tilde{\mathbf{h}}_{b,u}=\mathbf{h}_{b,u}-\hat{\mathbf{h}}_{b,u}$ is the vector of the CSI error. By exploiting the capacity bound of the fading channel with side information \cite[Sec. 2.3]{marzetta2016fundamentals}, the downlink achievable rate of User $u$ can be expressed as
	\begin{equation}\label{eqn1-6}
		\setlength{\abovedisplayskip}{2pt}%
		\setlength{\belowdisplayskip}{2pt}%
		\setlength{\abovedisplayshortskip}{0pt}%
		\setlength{\belowdisplayshortskip}{0pt}% 
		{{R}_{u}}={{\log }_{2}}\left( 1+\frac{ \displaystyle\frac{{{P}_{u}}}{\mathbb{E}\left\{ {{\left\| {{\bar{\mathbf{w}}}_{u}} \right\|}^{2}} \right\}}}{\sum\limits_{v=1}^{U}\mathbb{E}\left\{ {{\left| \sum\limits_{b=1}^{B}{\tilde{\mathbf{h}}_{b,u}^{\rm{H}}{{\mathbf{w}}_{b,v}}} \right|}^{2}} \right\}+\sigma _{n}^{2}} \right),
	\end{equation}
	where $\frac{{{P}_{u}}}{\mathbb{E}\left\{ {{\left\| {{\bar{\mathbf{w}}}_{u}} \right\|}^{2}} \right\}}$  represents the desired signal power, and ${\sum_{v=1}^{U}\mathbb{E}\left\{ {{\left| \sum_{b=1}^{B}{\tilde{\mathbf{h}}_{b,u}^{\rm{H}}{{\mathbf{w}}_{b,v}}} \right|}^{2}} \right\}}$  indicates the beamforming uncertainty power \cite{ngo2017total} due to the ZF precoding with recovered CSI. And the sum-rate can be expressed as 
	\begin{equation}\label{sum-rate-v1}
		\setlength{\abovedisplayskip}{2pt}%
		\setlength{\belowdisplayskip}{2pt}%
		\setlength{\abovedisplayshortskip}{0pt}%
		\setlength{\belowdisplayshortskip}{0pt}% 
		R_{\text {sum }}=\sum_{u=1}^U R_u.
	\end{equation}

	\section{Sum-Rate Analysis}
	\label{Sum-Rate Analysis}
	
	%\subsection{Sum-rate Expression for Port Selection}\label{Sum-rate Expression for Port Selection}

In this section, we derive an analytical expression of $R_{\text{sum}}$, which is needed for the port selection optimization and the performance evaluation.

	Denote ${\tilde{\bar {\bf h}}}_{{b,u}} = \mathbf{A}^{\rm{H}}_{{b,u}}{\tilde{\bar {\bf h}}}_{\Lambda_{b,u}} = \left[\tilde{\bar{h}}_{b,u,1},\ldots,\tilde{\bar{h}}_{b,u,M}\right]^{\text{T}}$ as the CSI error vector of overall port coefficients of BS $b$ for User $u$, and denote ${\hat{\bar {\bf h}}}_{{b,u}} = \mathbf{A}^{\rm{H}}_{{b,u}}{\hat{\bar{\mathbf{h}}}_{\Lambda_{b,u}}} = \left[\hat{\bar{h}}_{b,u,1},\ldots,\hat{\bar{h}}_{b,u,M}\right]^{\text{T}}$ as the recovered CSI vector of overall port coefficients of BS $b$ for User $u$.
	Therefore, the reconstructed CSI at the BS is
	\begin{equation}\label{feedback_channel}
		\begin{aligned}
			\hat{\mathbf{h}}_{b, u}
			& = \sqrt{M}\mathbf{F}_{\Lambda_{b,u}}\mathbf{B}_{\Lambda_{b,u}}{\hat{\bar {\bf h}}}_{{b,u}}
		\end{aligned}
	\end{equation}
	for the channel vector from BS $b$ to User $u$. 
	%\textcolor{red}{{This is a good approximation when our proposed feedback and reconstruction scheme in Section \ref{Section_feedback_reconstruction} has small feedback error. 	}}
	
From Eq. \eqref{channelmodel} and Eq. \eqref{feedback_channel}, the vector of the CSI error can then be expressed as
	\begin{equation}
		\tilde{\mathbf{h}}_{b, u}=\sqrt{M}\mathbf{F}_{\Lambda^C_{b,u}}\mathbf{B}_{\Lambda^C_{b,u}}\bar{\mathbf h}_{b,u} + \sqrt{M} {{\bf{F}}_{{\Lambda _{b,u}}}}{{\bf{B}}_{{\Lambda _{b,u}}}}{\tilde{\bar {\bf h}}}_{b,u}.
	\end{equation}
Define ${{\tilde{\mathbf{h}}}_{u}}={{\left[ \tilde{\mathbf{h}}_{1,u}^{{\mathrm{T}}},\ldots ,\tilde{\mathbf{h}}_{B,u}^{{\mathrm{T}}} \right]}^{{\mathrm{T}}}}$. For the interference power term in Eq. \eqref{eqn1-6}, we have 
	\begin{equation}
		\setlength{\abovedisplayskip}{2pt}%
		\setlength{\belowdisplayskip}{2pt}%
		\setlength{\abovedisplayshortskip}{0pt}%
		\setlength{\belowdisplayshortskip}{0pt}% 
		\begin{aligned}
			\mathbb{E}\left\{ {{\left| \sum\limits_{b=1}^{B}{\tilde{\mathbf{h}}_{b,u}^{\rm{H}}{{\mathbf{w}}_{b,v}}} \right|}^{2}} \right\}
			&\overset{(a)}{=} 
			{\mathbb{E}\left\{ |{{{\tilde{\mathbf{h}}}}^{\rm{H}}_{u}}\mathbf{{w}}_{v}|^2\right\}},\\
			&\overset{(b)}{=}
			\frac{\mathbb{E}\left\{ |{{{\tilde{\mathbf{h}}}}^{\rm{H}}_{u}}\mathbf{\bar{w}}_{v}|^2\right\}}{\mathbb{E}\left\{ {{\left\| {{{\mathbf{\bar{w}}}}_{v}} \right\|}^{2}} \right\}}{{P}_{v}}
		\end{aligned} 
	\end{equation} 
for $v \in \mathbb{U}$, 
	where $(a)$ follows from $ \sum\limits_{b=1}^{B}{\tilde{\mathbf{h}}_{b,u}^{\rm{H}}{{\mathbf{w}}_{b,v}}}=\tilde{\mathbf{h}}_{u}^{\rm{H}}{{\mathbf{w}}_{v}}$ and $(b)$ is from
	$ {{{\mathbf{{w}}}}_{v}} = \frac{\sqrt{P_v}{{{\mathbf{\bar{w}}}}_{v}}}
	{\sqrt{ \mathbb{E}\left \{ \left \| {{{\mathbf{\bar{w}}}}_{v}} \right \|^2  \right \}   }}$.
	Therefore, to derive $R_{\text {sum}}$, we need to derive  $\mathbb{E}\left\{ {{\left\| {{{{\bar{\mathbf w}}}}_{u}} \right\|}^{2}} \right\}$ and ${\mathbb{E}\left\{ |{{{\tilde{\mathbf{h}}}}^{\rm{H}}_{u}}\mathbf{\bar{w}}_{v}|^2\right\}}$.

		\begin{lemma}
		\label{lemma_0}
		% 意大利体
		\textit{${\hat {\mathbf{H}}^{\rm{H}}}\hat {\mathbf{H}}$ is a diagonal matrix and its $u$-th diagonal entry is
			\textcolor{black}{
		\begin{equation}\label{equal3_4}
			\left[\hat{\mathbf{H}}^{\mathrm{H}} \hat{\mathbf{H}}\right]_{u,u}= 
				M \sum_{b=1}^{B} \hat{\bar{\mathbf{h}}}_{b, u}^{\mathrm{H}} \mathbf{B}_{\Lambda_{b, u}}^{\mathrm{H}}\mathbf{B}_{\Lambda_{b, u}} \hat{\bar{\mathbf{h}}}_{b, u},\forall u \in \mathbb{U},
		\end{equation}}if no two users can share the same port from the same BS.}
	\end{lemma}
	\begin{proof}
	According to the definition of $\hat {\mathbf{H}}$, we have
		\begin{equation}
		\left[\hat{\mathbf{H}}^{\mathrm{H}} \hat{\mathbf{H}}\right]_{u,v}=\sum_{b=1}^{B}\hat{\mathbf{h}}^{\mathrm{H}}_{b,u}\hat{\mathbf{h}}_{b,v}, \forall u, v\in \mathbb{U}.
	\end{equation}	
	 From Eq. \eqref{feedback_channel},
	 \textcolor{black}{
	\begin{equation}\label{channel_feedback_inner_prod}
	\hat{\mathbf{h}}^{\mathrm{H}}_{b,u}\hat{\mathbf{h}}_{b,v}=M \hat{\bar{\mathbf{h}}}_{b, u}^{\mathrm{H}} \mathbf{B}_{\Lambda_{b, u}}^{\mathrm{H}}\mathbf{F}_{\Lambda_{b, u}}^{\mathrm{H}} \mathbf{F}_{\Lambda_{b, v}} \mathbf{B}_{\Lambda_{b, v}} \hat{\bar{\mathbf{h}}}_{b, v}.
	\end{equation}}Since ${{\Lambda }_{b,u}}\cap {{\Lambda }_{b,v}}=\emptyset ,\forall u\ne v\in \mathbb{U}$, we have 
\begin{equation}
	\linespread{1.5} \selectfont
\mathbf{F}_{\Lambda_{b, u}}^{\mathrm{H}} \mathbf{F}_{\Lambda_{b, v}} = \left\{ \begin{array}{cc}
	{\bf{I}}_{|\Lambda_{b,u}|}, &{\text{if }}\, u=v \\
	{\bf{0}}, &{\text{if }}\, u \neq v
	\end{array}
\right..
\end{equation}
This leads to Eq. \eqref{equal3_4} and the conclusion that ${\hat {\mathbf{H}}^{\rm{H}}}\hat {\mathbf{H}}$ is a diagonal matrix.
	\end{proof}
	
	 Lemma \ref{lemma_0} shows that via allocating each BS port to at most one user, i.e., for each $b\in \mathbb{B}$, the sets $\Lambda_{b,1},\ldots,\Lambda_{b,U}$ are mutually exclusive, the acquired multiple-BS channel vectors for the users $\hat{\mathbf{h}}_{u}$'s become orthogonal to each other. 
	 For the subsequent sum-rate analysis and port selection scheme, we consider this constraint ${{\Lambda }_{b,u}}\cap {{\Lambda }_{b,v}}=\emptyset,\forall u\ne v\in \mathbb{U}$ 
	  \footnote{\textcolor{black}{With a large number of antennas at the BS, the probability of overlapping between the strongest ports of different users tends to be small. By avoiding beam collisions, it is possible to select the strongest ports for users with the guarantee of inter-user channel orthogonality \cite[Sec. V-B]{alkhateeb2015limited}.}}.
	  
	  %By avoiding beam collisions, the system is likely to prefer strong ports or beams for transmission

	 	From ${\left\| {{{\mathbf{\bar{w}}}}_{u}} \right\|}^{2}={\left[ {{{\mathbf{\bar{W}}}}^{\rm{H}}}\mathbf{\bar{W}} \right]}_{u,u}$ we have
	 \begin{equation}\label{equal3_2}
	 	\setlength{\abovedisplayskip}{8pt}%
	 	\setlength{\belowdisplayskip}{8pt}%
	 	\setlength{\abovedisplayshortskip}{4pt}%
	 	\setlength{\belowdisplayshortskip}{4pt}% 
	 	\begin{aligned}
	 		{\left\| {{{{\bf{\bar w}}}_u}} \right\|^2} = {\left[ {{{\left( {{{\hat {\bf{H}}}^{\rm{H}}}\hat {\bf{H}}} \right)}^{ - 1}}} \right]_{u, u}}.
	 	\end{aligned}
	 \end{equation} 
	 Based on Lemma \ref{lemma_0}, this gives the following result:
	 \textcolor{black}{
	 \begin{equation}\label{equal3_5}
	 	\setlength{\abovedisplayskip}{8pt}%
	 	\setlength{\belowdisplayskip}{8pt}%
	 	\setlength{\abovedisplayshortskip}{4pt}%
	 	\setlength{\belowdisplayshortskip}{4pt}%
	 	{\left\| {{{{\bf{\bar w}}}_u}} \right\|^2} = \frac{1}{M}{\left( {\sum\limits_{b = 1}^B {\sum\limits_{k \in {\Lambda _{b,u}}} {{{\bar \beta }_{b,u,k}}{{\left| {{\hat{\bar h}_{b,u,k}}} \right|}^2}} } } \right)^{ - 1}}.
	 \end{equation}  }
	 
	 Define ${{\bf{B }}_{\Lambda_u}} = {\rm{blk}}[{{\bf{B }}_{\Lambda_{1,u}}},\ldots,{{\bf{B }}_{\Lambda_{B,u}}}]$, ${ {\hat{\bf{R}}_{\Lambda_u}}} = { {{\bf{R}}_{\Lambda_u}}} -{ {\tilde{\bf{R}}_{\Lambda_u}}} = \left(1-\varepsilon ^2\right){ {{\bf{R}}_{\Lambda_u}}}$, and we provide an exact analytical expression of $\mathbb{E}\left\{ {{\left\| {{{{\bar{\mathbf w}}}}_{u}} \right\|}^{2}} \right\}$ in the following Lemma \ref{lemma_1}.
	 
	 %  这部分添加长公式

	%{ \color{black}
		\begin{lemma}
			\label{lemma_1}
			% 意大利体
			\textit{Define 
				\textcolor{black}{
			\begin{equation}\label{S_u}
			{{\bf{S}}_u} \triangleq \left(1-\varepsilon ^2\right) {\bf{R}}_{\Lambda_u}^{{\rm{1}}/{2}}
			{{\bf{B }}_{\Lambda_u}}{{\bf{B }}^{\rm{H}}_{\Lambda_u}} {\bf{R}}_{\Lambda_u}^{{\rm{1}}/{2}}.
			\end{equation}}Denote the rank of ${{\bf{S}}_u}$ as $\rho_u$ and denote the positive eigenvalues of ${{\bf{S}}_u}$ as ${\lambda _{u,1}},\ldots,{\lambda _{u,{\rho _u}}}$.
			Given any port selection ${\Lambda }_{b,u}$'s, the average port power ${{\bar{\beta }}}_{b,u,l}$'s, and the port coefficient correlation  $\rho _{u,b,b'}^{l,l'}$'s, $\forall u \in \mathbb{U}$, $ b,b'\in \mathbb{B}$, $l\in \Lambda_{b,u} $, $l'\in \Lambda_{b',u} $, for ${\rho _u} > 1$, we have 
			\begin{equation}
				\label{equal_lemma1}
				\mathbb{E}\left\{ {{\left\| {{{\mathbf{\bar{w}}}}_{u}} \right\|}^{2}} \right\} = \frac{1}{M}\sum\limits_{k = 0}^\infty  {\frac{{{\alpha_{u,k}}}}{{2\beta_{u} \left( {\rho_{u}  + k - 1} \right)}}}
			\end{equation}
			for $u\in \mathbb{U}$, where 
			\begin{equation}
			\begin{small}
				\label{para_lemma1}
				\begin{array}{ll}
					{\beta _u} = \displaystyle \frac{{{\rho _u}}}{{2\sum\limits_{j = 1}^{{\rho _u}} {\lambda _{u,j}^{ - 1}} }}, & {\alpha_{u,0}} = \displaystyle \prod\limits_{j = 1}^{{\rho _u}} {\frac{{2{\beta _u}}}{{{\lambda _{u,j}}}}} ,  \\
					{b_{u,k}} = 2 \displaystyle\sum\limits_{j = 1}^{{\rho _u}} {{{\left( {1 - \frac{2 \beta_{u} }{{{\lambda _{u,j}}}}} \right)}^k}} , & {\alpha_{u,k}} = \displaystyle \frac{1}{{2k}}\sum\limits_{r = 0}^{k - 1} {{b_{u,k - r}}{\alpha_{u,r}}}, \forall k\ge 1.
				\end{array}
		\end{small}  	\end{equation} 
		}
		\end{lemma}
		\begin{proof}
			See Appendix A.		
		\end{proof}
		
		Next, we deal with ${\mathbb{E}\left\{ |{{{\tilde{\mathbf{h}}}}^{\rm{H}}_{u}}\mathbf{\bar{w}}_{v}|^2\right\}}$. 
		Define ${{\bf{B }}^{u}_{\Lambda_v}} = {\rm{blk}}[{\mathbf{A}}_{{1,v}} {\mathbf{B}}_{{1,u}},\ldots,{\mathbf{A}}_{{B,v}} {\mathbf{B}}_{{B,u}}]$.
		An approximate closed-form expression of ${\mathbb{E}\left\{ |{{{\tilde{\mathbf{h}}}}^{\rm{H}}_{u}}\mathbf{\bar{w}}_{v}|^2\right\}}$ is given in Lemma \ref{lemma2}.
		\begin{lemma}
			\label{lemma2}
				\textit{Define 
					\textcolor{black}{
			\begin{equation}\label{S_UV}
			\begin{tiny}
				{{\bf{S}}_{u,v}} \triangleq \left(1-\varepsilon ^2\right)
				{ {\bf{R}}_{\Lambda_v}^{1/2}}{\bf{B}}_{{\Lambda _v}}^u{\bf{B}}_{{\Lambda _v}}^{u,{\rm{H}}}
				{ {\bar{\bf{R}}^{u}_{\Lambda_v}}} 
				{{\bf{B}}_{{\Lambda _v}}}{\bf{B}}_{{\Lambda _v}}^{\rm{H}}{ {\bf{R}}_{\Lambda_v}^{1/2}},\\ 
			\end{tiny}
	\end{equation} 
where  
 \begin{equation}\label{S_UV_paras}
 	\linespread{1.2} \selectfont
 	{ {\bar{\bf{R}}^{u}_{\Lambda_v}}} \triangleq  \left\{ \begin{array}{cc}
 	\varepsilon ^2{ {{\bf{R}}_{\Lambda_u}}}, &{\text{if }}\, u=v \\
 		{ {{\bf{I}}_{K_v}}}, &{\text{if }}\, u \neq v
 	\end{array}
 	\right..
\end{equation}}Given any port selection ${\Lambda }_{b,u}$'s, the average port power ${{\bar{\beta }}}_{b,u,l}$'s, and the port coefficient correlation  $\rho _{u,b,b'}^{l,l'}$'s, $\forall u \in \mathbb{U}$, $ b,b'\in \mathbb{B}$, $l\in \Lambda_{b,u} $, $l'\in \Lambda_{b',u} $, for $\rho_v > 2$, we have 
			\textcolor{black}{
				\begin{equation}\label{lemma2_appro} 
					\begin{aligned}
						&{\mathbb{E}\left\{ |{{{\tilde{\mathbf{h}}}}^{\rm{H}}_{u}}\mathbf{\bar{w}}_{v}|^2\right\}} \approx \frac{{\rm{tr}}\left( {{\bf{S}}_{u,v}} \right)}{{{\left| {\rm{tr}}\left( {{\bf{S}}_{v}}  \right) \right|}^{2}}+\left\| {{\mathbf{S }}_{v}} \right\|_{\rm{F}}^{2}} + {{\delta }_{u,v}}{{\eta }_{v}},
					\end{aligned}
			 	\end{equation}}where ${{\delta }_{u,v}}$ is shown as Eq. \eqref{delta} at the bottom of the next page, 
				%${{\boldsymbol{\Sigma }}_{v,u,\text{s}}}$ is the eigenvalue matrix of ${{\bf{S}}_{v,u}}$  defined in Eq. \eqref{S_UV},
				and
				\setcounter{equation}{31}
				\textcolor{black}{\begin{equation}\label{eta}
					{{\eta }_{v}}=\frac{1}{M^2}\sum\limits_{k=0}^{\infty }{\frac{{{\alpha}_{v,k}}}{{{\left( 2{{\beta }_{v}} \right)}^{2}}\left( {{\rho }_{v}}+k-1 \right)\left( {{\rho }_{v}}+k-2 \right)}},\,{\text{for}}\,\, \rho_v > 2.
				\end{equation}}with ${{\alpha }_{v,k}}$, $\beta_v$ and $\rho_v$ being defined in Eq. \eqref{para_lemma1}.
			}
		\end{lemma}
		\begin{proof}
			See Appendix B.
		\end{proof}
	
		\newcounter{TempEqCnt2}                         % 创建临时变量TempEqCnt
		\setcounter{TempEqCnt2}{\value{equation}} % 将当前公式序号 赋给TempEqCnt
		\setcounter{equation}{30}  % 当前公式序号变为x，x等于长公式应有的序号减1.
		\begin{figure*}[hb] % hb底部，ht为头部
			\centering % 公式居中
			\hrulefill % 添加一条水平线
			\vspace*{0pt} % 调整线与公式之间的距离
			\textcolor{black}{
				\begin{equation}\label{delta} 
					{{\delta }_{u,v}} =
					\begin{cases}
						\vspace{1em}
						M^2\sum\limits_{b' = 1}^B {\sum\limits_{b \ne b'}^B {\sum\limits_{l \in {\Lambda _{b,u}}} {\sum\limits_{l' \in {\Lambda _{b',u}}} {
										{
											\varepsilon^2\left(1- \varepsilon ^2\right){\rho}_{u,b,b'}^{l,l',2} {{\bar \beta }_{b,u,l}}{{\bar \beta }_{b',u,l'}} }}}}} ,
						&\hspace{1em}{\textit{ if }} \, v=u \\
						M^2\sum\limits_{b' = 1}^B {\sum\limits_{b \ne b'}^B {\sum\limits_{l \in {\Lambda _{b,v}}} {\sum\limits_{l' \in {\Lambda _{b',v}}} \left(1- \varepsilon ^2\right){\rho _{u,b,b'}^{l,l'}{\rho}_{v,b,b'}^{l,l'} \sqrt {{{\bar \beta }_{b,v,l}}{{\bar \beta }_{b,u,l}}{{\bar \beta }_{b',v,l'}}{{\bar \beta }_{b',u,l'}}} } } } },
						&\hspace{1em}{\textit{ if }} \, v \ne u. \\
					\end{cases}
			\end{equation} }
		\end{figure*}
		\setcounter{equation}{\value{TempEqCnt2}} % 把TempEqCnt中存的公式序号赋回给当前公式序号

		\begin{theorem}\label{Theorem1}
			\textit{Given any port selection ${\Lambda }_{b,u}$'s, the average port power ${{\bar{\beta }}}_{b,u,l}$'s, and the port coefficient correlation  $\rho _{u,b,b'}^{l,l'}$'s, $\forall u \in \mathbb{U}$, $ b,b'\in \mathbb{B}$, $l\in \Lambda_{b,u} $, $l'\in \Lambda_{b',u}$, for $\rho_v > 2$, the downlink achievable rate of user $u$ has the approximation as shown in Eq. \eqref{equal_sumrate} at the bottom of the next page,
			where ${\mu _v} =  \sum\limits_{k = 0}^\infty  {\frac{{{\alpha_{v,k}}}}{{2{\beta _v}\left( {{\rho _v} + k - 1} \right)}}} $.}
		\end{theorem}
	%\vspace{0.2cm}
		\begin{proof}
			Via substituting Eq. \eqref{equal_lemma1} and Eq. \eqref{lemma2_appro} into Eq. \eqref{eqn1-6}, Eq. \eqref{equal_sumrate} is readily obtained.
		\end{proof}
		
	\newcounter{TempEqCnt}                         % 创建临时变量TempEqCnt
	\setcounter{TempEqCnt}{\value{equation}} % 将当前公式序号 赋给TempEqCnt
	\setcounter{equation}{32}  % 当前公式序号变为x，x等于长公式应有的序号减1.
	\begin{figure*}[hb] % hb底部，ht为头部
		\centering % 公式居中
		%\hrulefill % 添加一条水平线
		\vspace*{-18pt} % 调整线与公式之间的距离
		\textcolor{black}{
			\begin{equation}
				\label{equal_sumrate}
				%	\begin{aligned}
					{{R}_{u}}\approx {{\log }_{2}}\left( 1+ \frac{{{P}_{u}}}{{{\mu }_{u}}}
					{{ \left( \sum\limits_{v =1}^{U}{\frac{{{P}_{v}}}{{{\mu }_{v}}}\left( \frac{{\rm{tr}}({\mathbf{S}_{u,v}})}{{{\left| {\rm{tr}}({\mathbf{S}_{v}}) \right|}^{2}}+\left\| {{\mathbf{S }}_{v}} \right\|_{\rm{F}}^{2}}+{{\delta }_{u,v}}{{\eta }_{v}} \right)+M{\sigma _n^2}} \right) }^{-1}} \right).
					%	\end{aligned}
		\end{equation} }
	\end{figure*}
	\setcounter{equation}{\value{TempEqCnt}} % 把TempEqCnt中存的公式序号赋回给当前公式序号
		
		% 计算公式(21)所需要的复杂度分析一下。
		The closed-form expression in Eq. \eqref{equal_sumrate} allows efficient evaluation of the system sum-rate $R_\text{sum}$ given arbitrary transmit power $P_u$'s, the average channel power in different port directions, and the correlation among different port coefficients.

	The computational complexity involved in the sum-rate calculation using Eq. \eqref{equal_sumrate} is analyzed as follows.
	The complexity using eigen-decomposition of ${\bf{S}}_v$ is typically on the order of $\mathcal{O}(K_v^3)$.
	The computation of ${{\delta }_{u,v}}$ has a complexity of the order $\mathcal{O}(K_v^2)$.
	For the computation of ${{\eta }_{v}}$, a truncated summation of the first $L_{\mu}$ terms can be employed, along with the normalization of ${{\alpha}_{v,k}}$’s. 
	The additional complexity required for ${{\eta }_{v}}$ is on the order of $\mathcal{O}({L^2_{\mu}})$.
	The complexity for the computation of $\mu_{v}$ is the same as that of  $\eta_{v}$. 
	For the special case that the total number of ports allocated to each user is the same, i.e., $K_v = P$, $\forall v \in \mathbb{U}$, to calculate the achievable rate using 
	Eq. \eqref{equal_sumrate}, the computational complexity has the order of $\mathcal{O}(UP^3)+
	\mathcal{O}(U{L^2_{\mu}})$. The complexity of calculating the system sum-rate is thus $\mathcal{O}(U^2P^3)+
	\mathcal{O}(U{L^2_{\mu}})$.

\section{Port Selection Design Schemes}
\label{Port Selection Design Schemes}
With the derived analytical sum-rate expression given in Eq. \eqref{equal_sumrate}, the port selection optimization problem 
can be formulated as 
\setcounter{equation}{33}  
\begin{eqnarray}\label{eqn_opt}
	&\hspace{-5cm}\underset{\left\{\Lambda_{b, u}, b\in\mathbb{B},u \in \mathbb{U} \right\}}{\max } R_{\text {sum }} \\
	\hspace{-2.2cm}\text{s.t.}&\hspace{-2.2cm} 
	\sum_{b=1}^B\left|\Lambda_{b, u}\right| \leq N_u, \forall u \in \mathbb{U},\nonumber\\
	&\Lambda_{b,u}\cap\Lambda_{b,u'}=\emptyset, \forall b\in \mathbb{B} \text{ }\&\text{ } u\neq u'\in \mathbb{U} \nonumber,
\end{eqnarray} 
where in the first constraint $N_u$ is the maximum number of ports that User $u$ can be assigned, and in the second constraint no two users can share the same port from the same BS.

	\subsection{Greedy-Search-Type Port Selection}\label{Greedy-Search-Type Port selection}
	
	One way to find the optimal solution to the port selection problem presented in Eq. \eqref{eqn_opt} is the exhaustive search but with a prohibitively high computational cost. In this work, we propose an alternative approach, namely the GS-JPS algorithm, which can be used to find a feasible suboptimal solution.		
	% 顺序是先选择更新哪个用户，再决定更新该用户的哪个服务基站，最后再决定选择该基站服务该用户的哪个端口。
	% 原因是：因为是以和速率最大化为目标，速率高的用户更新优先级高，因为这些用户的端口更新对和速率影响较大；给定用户的情况下，基站到用户的总端口能量越大，在不考虑用户间端口冲突的前提下，越有可能给用户提供比较好的速率性能。在给定用户、基站的情况下，同样平均能量高的端口对速率的影响相对较大，应该优先考虑。
	%\vspace{-1cm}
	%\vspace{-0.5cm}
	\begin{algorithm}[htb!] 
		%\SetAlgoNoLine %控制有无竖线
		%\setstretch{1.48}
		\linespread{1} \selectfont
		\caption{Greedy-Search Based Joint Port Selection (GS-JPS)}  
		\label{algorithm1}
		\KwIn{${\bar \beta}_{b,u,m}$'s, $\mathbf{R}_{u}$'s, $P_u$'s, $\sigma_{n}^{2}$, $N_{\text {rand}}$, $|\Lambda_{b, u}|$'s, $R_\text{sum}= 0$.} 
		\KwOut{$\Lambda^{\star}_{b, u}$'s.} 
		
		\For{$n =1,\ldots,N_{\textnormal{rand}} $} 
		{ {\textbf{Initialization: }$\mathbf{v} = \text{randperm}(U)$. ${\Lambda }_{b}=\emptyset$, $\Lambda^C_{b} = \mathbb{M}-{{\Lambda }_{b}}$, $\forall b\in \mathbb{B}$;
				
				\For{$u =1,\ldots,U$}
				{Obtain indices of $|\Lambda_{b, v_u}|$ strongest ports among $\Lambda^C_{b}$ to form $\Lambda_{b, v_u}$, $\forall b\in \mathbb{B}$;
					
					${{\Lambda }_{b}}={{\Lambda }_{b}} {\mathop{\bigcup }}{{\Lambda }_{b,v_u}}$, $\Lambda^C_{b} = \mathbb{M}-{{\Lambda }_{b}}$, $\forall b\in \mathbb{B}$;
				}
				
				Calculate $R^{{(n)}}_{u}$'s and $R_{\text{sum}}^{{(n)}}$ using Eq. \eqref{equal_sumrate};
				
				Obtain $\bar{\mathbf{v}}$ satisfying $R^{{(n)}}_{\bar{v}_1} \ge R^{{(n)}}_{\bar{v}_2}\ldots\ge R^{{(n)}}_{\bar{v}_U}$\;}	
			\For{$u =1,\ldots,U$}
			{{Calculate $\mathbf{b}$ for user $\bar{v}_u$ satisfying $\sum_{m=1}^M{\bar \beta}_{b_1,\bar{v}_u,m} \ge \ldots \ge \sum_{m=1}^M{\bar \beta}_{b_B,\bar{v}_u,m}$\;}
				
				\For{$i=1,\ldots,B$}
				{ {Put elements of ${\Lambda }_{b_i,\bar{v}_u}$ with decreasing port average power into the vector $\bar{\mathbf{p}}$\;}	
					\For{$m =1,\ldots,{|{\Lambda }_{b_i,\bar{v}_u}|}$}
					{ {Replace port $\bar{p}_m$ with port $l$ and calculate $R_{\text{sum},l}$,  $\forall l \in  \Lambda^C_{b_i}$\;}
						{$l^{\star} = \arg \max_{l \in  \Lambda^C_{b_i}} \left\{R_{\text{sum},l}\right\}$\;}
						\If{${R_{\textnormal{sum},l^{\star}}}>{{R}_{\textnormal{sum}}^{{(n)}}}$}
						{
							{Conduct ${{R}_{\text{sum}}^{{(n)}}} \gets R_{\text{sum},l^{\star}}$;
								
								Update ${{\Lambda }_{b_i,\bar{v}_u}}\gets ({\Lambda _{b_i,{\bar v}_u}}-\left\{ \bar{p}_m \right\}){\mathop{\bigcup }}\{l^{\star}\}$ and $\Lambda^C_{b_i} \gets (\Lambda^C_{b_i} - \{l^{\star}\}) {\mathop{\bigcup }}\left\{ \bar{p}_m \right\}$\;
							}	
						}
					}
				}	
			}
			\If{${R_{\textnormal{sum}}^{{(n)}}}>{{R}_{\textnormal{sum}}}$}
			{
				{Conduct ${{R}_{\text{sum}}} \gets R_{\text{sum}}^{{(n)}}$\;}
				{Update ${{\Lambda }^{\star}_{b,u}} = {{\Lambda }_{b,u}}$, $\forall u \in \mathbb{U}, b \in \mathbb{B}$\; }
			}	
		}
	\end{algorithm}
	 %\vspace{-1cm}	

The proposed port selection is shown in Algorithm \ref{algorithm1}, which contains $N_{\text{rand}}$ rounds of port selection and update. For the $n$-th round, first, a random initialization is conducted to determine the user order for port selection, i.e., $\mathbf{v} = \text{randperm}(U)$. Then, the CU selects ports for users $v_u, u=1,\ldots,U$ in turn, where $v_u$ is the $u$-th element of the vector $\mathbf{v}$. Specifically, the CU selects for User $v_1$ the $|\Lambda_{b, v_1}|$ ports of BS $b$ with the highest average power, and the remaining set of unselected ports of BS $b$ is  $\Lambda^C_{b} = \mathbb{M}-\Lambda_{b, v_1}$. Then for User $v_2$, the CU selects $|\Lambda_{b, v_2}|$ ports of BS $b$ with the highest average power from $\Lambda^C_{b}$, and the remaining set of unselected ports of BS $b$ becomes $\Lambda^C_{b} = \Lambda^C_{b}-\Lambda_{b, v_2}$. This process is repeated until the port selection for User $v_U$ is finished.	
	After the initial port selection, the per-user rate, denoted as  $R^{{(n)}}_{u}, u=1,\ldots,U$ and the sum-rate $R_{\text{sum}}^{{(n)}}$ are calculated from Eq. \eqref{equal_sumrate}. 
	Then we order the users in the descending order of their individual rates  to obtain $\bar{\mathbf{v}}$ satisfying 
	\begin{equation}\label{user_oder}
	R^{{(n)}}_{\bar{v}_1} \ge R^{{(n)}}_{\bar{v}_2}\ldots\ge R^{{(n)}}_{\bar{v}_U}.
	\end{equation} 
	For each user ${\bar{v}_u}, u=1,...,U$, we find the order of the BSs  $\mathbf{b}$ satisfying
	\begin{equation}\label{BS_oder}
		\sum_{m=1}^M{\bar \beta}_{b_1,\bar{v}_u,m} \ge \sum_{m=1}^M{\bar \beta}_{b_2,\bar{v}_u,m} \ldots \ge \sum_{m=1}^M{\bar \beta}_{b_B,\bar{v}_u,m},
	\end{equation}
and for given user ${\bar{v}_u}, u=1,\ldots,U$ and BS $b_i, i=1,\ldots,B$, we order the ports in $\Lambda_{b_i, \bar{v}_u}$ in the decreasing order of the port average power.
	Sequentially based on the ordering, each port in $\Lambda_{b_i,\bar{v}_u}$ is replaced by each port in $\Lambda_{b_i}^C$. 
	The sum-rate resulting from the substitution with port $l$ is recorded as ${R_{\textnormal{sum},l}}$, and $l^{\star}$ denotes the replacement port that results in the highest sum-rate.
	If ${R_{\textnormal{sum},l^{\star}}} > R^{(n)}_{\text {sum}}$, we update $\Lambda_{b_i, \bar{v}_u}$,  $\Lambda^C_{b_i}$ and $R^{(n)}_{\text {sum}}={R_{\textnormal{sum},l^{\star}}} $.
	After $N_{\text{rand}}$ rounds of port selection initialization and update, the port selection with the highest sum-rate is adopted as the algorithm output.

	For the updating of port selection, it is critical to determine the order and condition under which an update is to be made. 
	Since the optimization problem aims to maximize the sum rate, we first use Eq. \eqref{user_oder} to obtain the order of users for port selection updating, giving the user with a higher rate a higher update priority.
	Then, the order of the BSs is determined according to Eq. \eqref{BS_oder} where the BS with higher total port energy has a higher priority since it provides the most useful signal power for the given user with high probability.
	In addition, for the given user and BS, the port with higher average energy also has a greater impact on the user rate and needs to be prioritized for updating.
	Under the aforementioned  ordering and updating strategy, the GS-JPS algorithm can update the port selection in a greedy but relatively efficient manner.

	For the computational complexity analysis of the proposed GS-JPS algorithm, we illustrate the special case that each BS assigns the same number of ports to each user, i.e., ${\left| {{\Lambda }_{b,v}} \right|} = T$, $\forall b \in \mathbb{B}$, $\forall v \in \mathbb{U}$. 
	For one round of iteration, the algorithm needs to sort the users, BSs, and ports with the complexity of $\mathcal{O}\left(U^2+UB^2+UBT^2  \right)$.
	Additionally, the calculation of the sum-rate takes $1+UBT(M-UT)$ times each with the complexity of $\mathcal{O}\left(U^2B^3T^3+U{L^2_{\mu}}\right)$
	Therefore, the total computational complexity of the algorithm with $N_{\text{rand}}$ rounds of iteration is 
	$\mathcal{O}\left({N_{\text{rand}}}U^2BT(M-UT)(UB^3T^3+
	{L^2_{\mu}})\right)$
	\footnote{Note that the joint port selection algorithm is executed once in one channel statistical coherence duration, i.e., no additional computational complexity is introduced for the majority of channel coherence time slots contained in a channel statistical coherence duration.}.
	\textcolor{black}{The statistical SLNR maximizing port selection (SLNR-PS) scheme proposed in \cite{kim2020downlink} also enables joint port selection with the computational complexity of $\mathcal{O}\left(U^2M^2B^3(M+UL)(T^3+L^3)\right)$. With the increase of the number of BS antennas $M$, our proposed algorithm can obtain port selection results with lower complexity compared to SLNR-PS.}

\subsection{DL-Assisted Fast Port Selection}\label{ML-assisted enhanced port selection}

In cell-free massive MIMO systems, the complexity of the GS-JPS algorithm proposed in Section \ref{Greedy-Search-Type Port selection} becomes prohibitive for scenarios with a large number of antennas and users. Moreover, when dealing with rapidly changing channel statistics, this issue becomes even more pressing.  Therefore, we further propose the DL-JPS scheme in the following where a deep neural network (DNN) is used to simulate the decision mechanism of the GS-JPS algorithm and realize fast online port selection.

The DNN is trained in a supervised way to learn the mapping from the normalized port average power to the port selection labels provided by the proposed GS-JPS algorithm. 
We formulate the port selection as a classification problem. Specifically, the average port power of User $u$ corresponding to BS $b$ is denoted as $\bar {\boldsymbol{\beta }} _{b,u} = \left[ {\bar \beta  _{b,u,1}, \ldots ,\bar \beta  _{b,u,M}} \right]^{\rm{T}} \in {\mathbb{R}^{M \times 1}}$. Then,  ${\bar {\boldsymbol{\beta }}} = {\left[ {\bar {\boldsymbol{\beta }} _1, \ldots ,\bar {\boldsymbol{\beta }} _U} \right]} \in \mathbb{R}^{BM \times U}$ is the DNN input with $\bar {\boldsymbol{\beta }} _u = {\left[ {\bar {\boldsymbol{\beta }} _{1,u}^{{\rm{T}}}, \ldots ,\bar {\boldsymbol{\beta }} _{B,u}^{{\rm{T}}}} \right]^{\rm{T}}} \in \mathbb{R}^{BM \times 1}$.
As shown in Fig. \ref{figure4_3}, to efficiently represent the selected ports of all collaborative BSs for all serving users, the $M\times B\times U$-dimensional output of the DNN is divided into $U\times B$ blocks. For the $\left( u,b \right)$-th block, $\forall u\in \mathbb{U}$, $\forall b\in \mathbb{B}$, $M$ logistic activation function-based binary classifiers are adopted to determine whether each port is selected by BS $b$ for User $u$.
Given the output of the $\left( u,b \right)$-th block classifier, i.e., ${{\hat{\mathbf{p}}}_{b,u}}=\left[ {\hat{p}_{b,u,1}},\ldots ,{\hat{p}_{b,u,M}} \right]$, the indices of selected ports are determined, i.e., 
$
{\hat{{\Lambda }}_{b,u}}=f_{\text{find}}^{\left| {{\Lambda }_{b,u}} \right|}\left( {\hat{\mathbf{p}}_{b,u}} \right),
$
where the function $f_{\text{find}}^{\left| {{\Lambda }_{b,u}} \right|}\left( \cdot \right)$ outputs the set of subscripts of the top  $\left| {{\Lambda }_{b,u}} \right|$ maximums in the input vector.
The cross entropy error (CEE) is used as the loss function, i.e.,
\begin{equation}\label{loss1}
\begin{aligned}
	{{L}_{\text{loss}}}=-\sum\limits_{b=1}^{B}\sum\limits_{u=1}^{U}\sum\limits_{m=1}^{M}&{{p}_{b,u,m}}\log {{\hat{p}_{b,u,m}}}\\&+(1-{{p}_{b,u,m}})\log ({1-{\hat{p}_{b,u,m}}}),
\end{aligned}
\end{equation} 
where ${p}_{b,u,m}$ is the label with its value being $1$ if the $m$-th port is selected by BS $b$ for User $u$, i.e., $m \in {\Lambda _{b,u}}$, and $0$ otherwise. 
\begin{figure}[htp]
	%\vspace{-1em}
	%\setlength{\abovecaptionskip}{0.1cm}   %调整图片标题与图距离
	\centering 
	\includegraphics[width=1.05\linewidth]{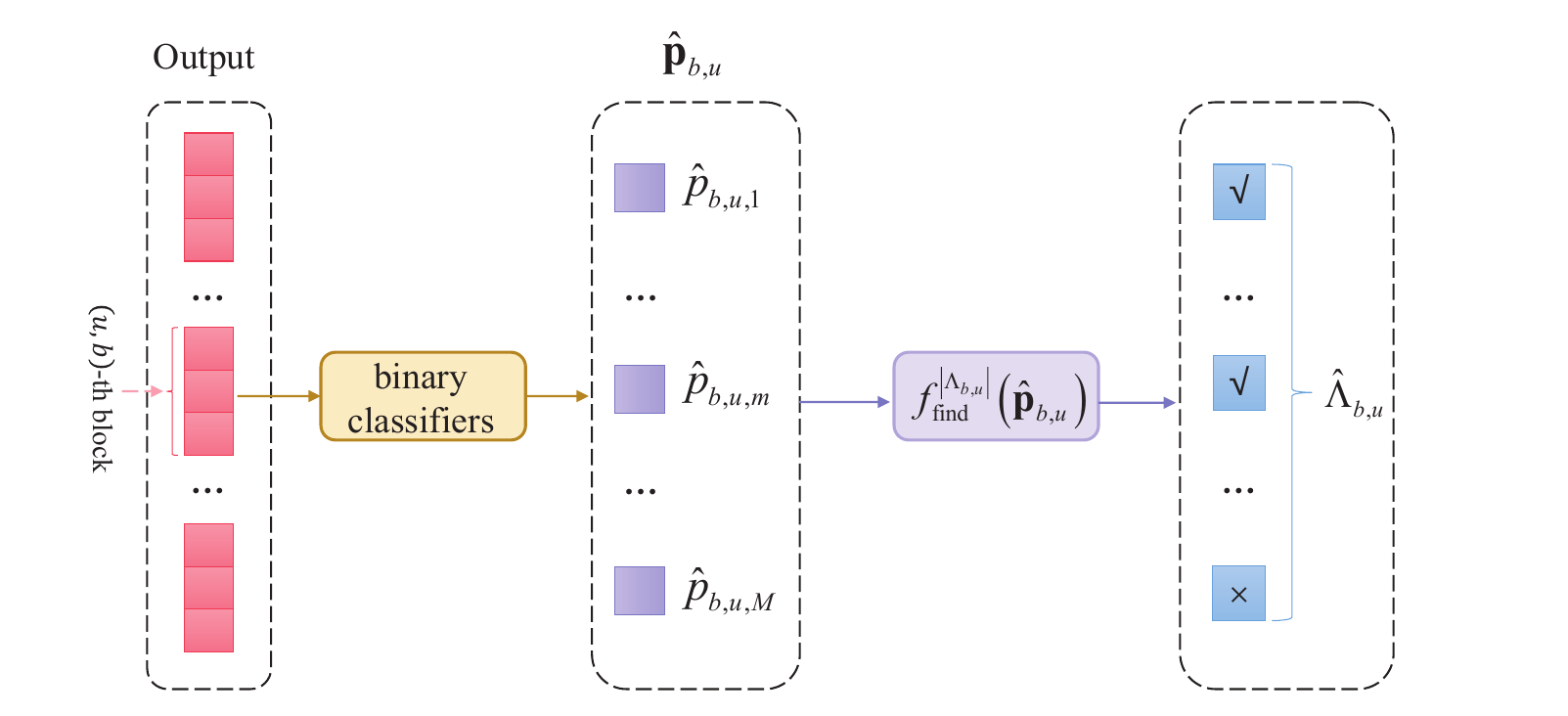} 
	\caption{The structure of the DNN output layer.}
	\label{figure4_3}
	%\vspace{-2em}
\end{figure}

	\section{Simulation and Discussion}
	\label{Simulation and Discussion}

	% 基线：SLNR方法。
	% 图1:验证推导速率指标的正确性 
	In this section, numerical results are provided to demonstrate the performance of the proposed GS-JPS, DL-JPS, and EDT-based feedback algorithms. For comparison, we employ the SLNR-PS scheme and the maximum magnitude selection scheme (MM-S)  \cite{sayeed2013beamspace} as the baselines. 
%	Both baseline schemes also incorporate the principle that "the same beam from a single BS will not be shared among users".	
	\subsection{Simulation Setup}
	%\subsubsection{The key channel and scenario parameters}
		
	\begin{figure}[htp]
		%	\setlength{\abovecaptionskip}{0.1cm}   %调整图片标题与图距离
		%\vspace{-1em}
		\centering 
		\includegraphics[width=0.8\linewidth]{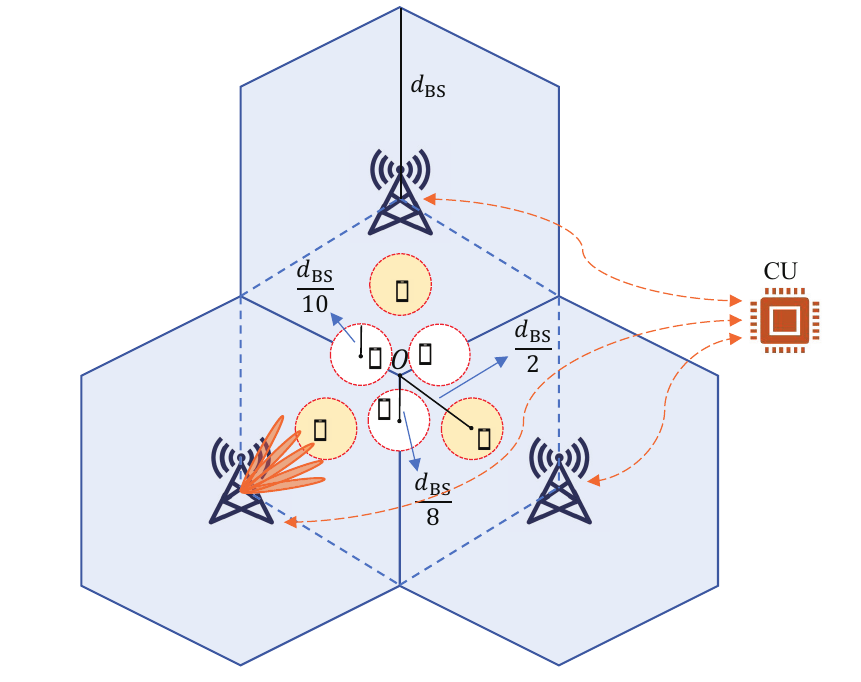}
		\caption{Schematic diagram of simulation scenario.}
		\label{figure0}
		%\vspace{-2em}
	\end{figure}

	\newcounter{TempEqCnt3}                         % 创建临时变量TempEqCnt
	\setcounter{TempEqCnt3}{\value{equation}} % 将当前公式序号 赋给TempEqCnt
	\setcounter{equation}{40}  % 当前公式序号变为x，x等于长公式应有的序号减1.
	\begin{figure*}[hb] % hb底部，ht为头部
		\centering % 公式居中
		\hrulefill % 添加一条水平线
		\textcolor{black}{
			\begin{equation}\label{Covariance_paradigm2}
				\mathbf{R}_{u,b,b' } = 
				\begin{bmatrix}
					\mathbf{0}_{(p_b-1)\times (p_{b'}-1)}& \mathbf{0}_{(p_b-1)\times L_0}& \mathbf{0}_{(p_b-1)\times (M-L_0-p_{b'}+1)}\\
					\mathbf{0}_{L_0\times (p_{b'}-1)}& \rho_{\rm{c}} \mathbf{I}_{L_0\times L_0}& \mathbf{0}_{L_0\times (M-L_0-p_{b'}+1)}\\
					\mathbf{0}_{(M-L_0-p_{b}+1)\times (p_{b'}-1)}& \mathbf{0}_{(M-L_0-p_{b}+1)\times L_0}&\mathbf{0}_{(M-L_0-p_{b}+1)\times (M-L_0-p_{b'}+1)}
				\end{bmatrix}, {\text{ for }} \, b' \ne b.
		\end{equation} }
		\vspace*{0pt} % 调整线与公式之间的距离
		%\hrulefill % 添加一条水平线
	\end{figure*}
	\setcounter{equation}{\value{TempEqCnt3}} % 把TempEqCnt中存的公式序号赋回给当前公式序号

	Our simulations focus on a cell-free massive MIMO system operating in the Urban Microcell scenarios with $B=3$ BSs equipped with $M=64$ antennas individually. 
	 As shown in Fig. \ref{figure0}, the polar coordinate of the center point $O$ is $(0,0)$, and the adjacent service areas of the BSs are connected to form a hexagonal region with inter-site spacing $d_{\rm{BS}}=250$ m. $U=6$ users are served by the BSs via the collaborative transmission among which $3$ are inter-cell users and $3$ are cell-edge users. 
	 To model the users' random location in small local areas, each user location is generated following a uniform distribution on a circle with radius $r_0 = d_{\rm{BS}}/10$. The centers of the circles for the intra-cell users and the cell-edge users are set as $\left( {d_{\rm{BS}}}/{2},-{\pi}/{6}+2(b-1){\pi}/{3}\right)$, and $\left({d_{\rm{BS}}}/{8},{\pi}/{6}+2(b-1){\pi}/{3}\right),b\in\mathbb{B}$, respectively. 
	The coefficient ${\bar \beta}_{b,u,m}$ of each path/port in the adopted channel model in Eq. \eqref{channelmodel} satisfies $\sum\nolimits_{m=1}^{M}{\bar \beta}_{b,u,m} = {\bar \beta}_{b,u}$, where the channel power angular spectrum over different ports follows the truncated Laplacian distribution \cite{pedersen2000stochastic}. ${\bar \beta}_{b,u}$ can be expressed as
	\begin{equation}
		{{\bar{\beta }}_{b,u}}[{\rm{dB}}]=- 28-20\, {\text{log}_{10}}\left( f_0 \right)-22\,{\text{log}_{10}}\left( {{d}_{b,u}} \right) [\rm{dB}] ,
	\end{equation}
	where $f_0 = 2.1\,$GHz is the downlink transmission frequency, and ${d_{b,u}}$ (in meter) denotes the distance between User $u$ and BS $b$. The strategy of equal power allocation among users is adopted under the single-BS power constraint. The maximum power of one BS is $P_{\text{tx}}$, and the system SNR  is defined as $\text{SNR} = {P_{\text{tx}}{\bar\beta}}/{\sigma_n^2}$ with ${\bar\beta} = \min_{b\in\mathbb{B},u\in\mathbb{U}} \{ \bar\beta_{b,u}\}$ being the largest path loss of BS-user link and $\sigma_n^2 = 1$ for noise power normalization. For example, for $\text{SNR} = 15$ dB, $P_{\text{tx}} = 95$ dB if ${\bar\beta}=-80$ dB.
	
% 用户侧SNR？用户是哪一个？边缘圆圈中心用户？基站是哪个？离他最近的基站？怎么定义的？
% P_u * \bar{\beta}_{b,u}/接收机噪声功功率（sigma^2_n = ?）

%When the power constraint of BS is $100\,$dB, the SNR of cell edge users is about $15\,$dB. 

	\begin{figure*}[ht]
		\vspace{-0.5cm}  % 调整与上文的间距
		\centering 
		\includegraphics[width=0.8\linewidth]{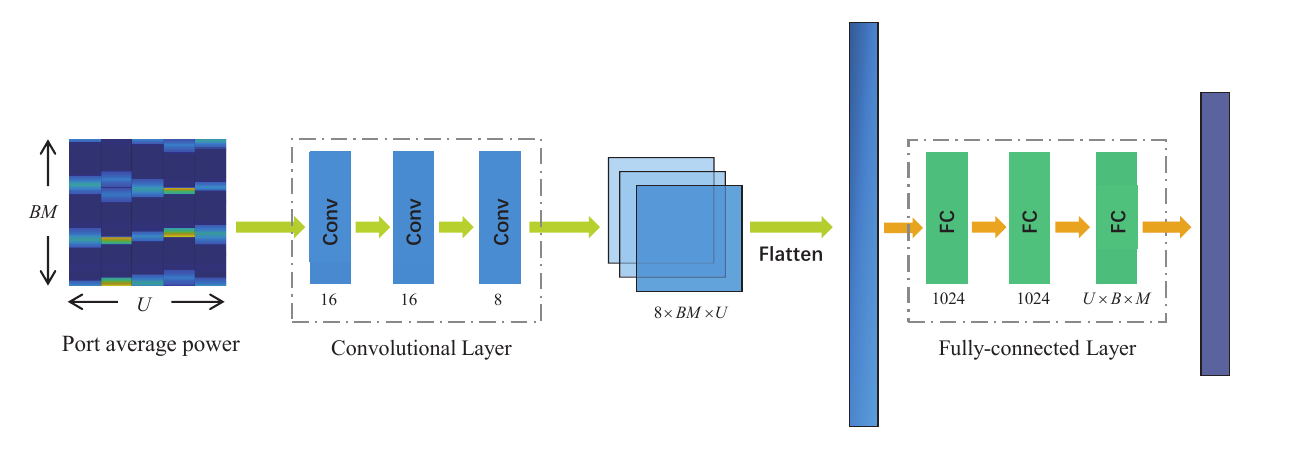}
		\caption{The architecture of DL-assisted fast port selection network.}
		\label{DNN}
		\vspace{-0.5cm}
	\end{figure*}

	Denote the port with non-zero average power as the effective channel port. The effective channel ports in Eq. \eqref{channelmodel} are evenly distributed on both sides of the line-of-sight (LoS) path between the BS and the user, and the number of effective ports is $L = 20$ for typical angular spread (AS) value of $AS=18{}^\circ$. 
	\textcolor{black}{For the modeling of the port-domain correlation matrix, we assume the matrix for User $u$ has the following form
	\begin{equation}\label{Covariance_paradigm}
		\mathbf{R}_{u} = \begin{bmatrix}
			\mathbf{R}_{u,1,1} &\dots   &\mathbf{R}_{u,1,B} \\
			\vdots &  \ddots & \vdots\\
			\mathbf{R}_{u,B,1}&\dots   &\mathbf{R}_{u,B,B}
		\end{bmatrix},
	\end{equation}
	where $\mathbf{R}_{u,b,b'}$ denotes the mutual covariance matrix of ${{{\bar{\mathbf h}}}_{b,u}}$ and ${{{\bar{\mathbf h}}}_{b',u}}$. The exponential correlation model \cite{loyka2001channel}, \cite{zhang2023interleaved} is considered for the covariance matrix of the same BS, i.e., 
	\begin{equation}
		\linespread{1.2} \selectfont
		\left [ \mathbf{R}_{u,b,b} \right ]_{l,l'} = \left\{ \begin{array}{cc}
			{\rho}^{\left |l-l'  \right | }_{\rm{s}}, &{\text{if }} l,l' {\text{ are effective ports}}   \\
			0, &{\text{otherwise}}
		\end{array}
		\right.,
	\end{equation}
	where $\rho_{\rm{s}}$ is the correlation coefficient of neighboring ports, which satisfying $\left | \rho_{\rm{s}} \right | \le 1 $. }
	%For the modeling of the port-domain correlation matrix, we assume that for $u\in \mathbb{U}$ and $b\in \mathbb{B}$, $\rho_{u,b,b}^{l,l} = 1$ for $l\in\mathbb{M}$ and $\rho_{u,b,b}^{l,l'} = 0$ for $l\neq l'\in\mathbb{M}$. This is because large-scale antenna arrays at the BS provide high spatial resolution, making different port coefficients at the same BS correspond to different distinguishable scatterers.
	 For the different ports from different BSs, we assume that there is at most one non-zero entry in $\rho_{u,b,b'}^{l,l'}$'s for all $l'\in\mathbb{M}$ with given $l\in\mathbb{M}$. For simplicity, we assume that the first $L_0$  effective ports of User $u$ at each BS are with inter-BS correlation, and 
	 \textcolor{black}{the correlation coefficient $\rho_{u,b,b'}^{l,l'}$ is set to $\rho_{\rm{c}}$. For example, if the indices of $L$ consecutive effective ports from BS $b$ and  $b' \ne b$ to User $u$ are $\{p_b,\ldots,L+p_b\}$ and $\{p_{b'},\ldots,L+p_{b'}\}$, respectively, we have $\mathbf{R}_{u,b,b'}$ as shown in Eq. \eqref{Covariance_paradigm2} at the bottom of this page.}
	In addition, we set $K_u = P$, $\forall u \in \mathbb{U}$ in the proposed GS-JPS method. 
	And the constraint of no port sharing between users is considered.

The DNN architecture for the proposed DL-JPS algorithm is shown in Fig. \ref{DNN}. It 
consists of $3$ convolutional layers (kernel size: $3\times3$), $1$ reconstruction layer and $3$ fully connected layers. The number of kernels for $3$ convolutional layers is $16$, $16$, and $8$, respectively. The reconstruction layer reorganizes the convolved structure into a $1$-dimensional form to facilitate subsequent network processing. The following $3$ fully connected layers are with $1024$ Relu neurons, $1024$ Relu neurons, and $U\times B\times M$ Sigmoid neurons, respectively. In addition,  batch normalization (BN) and the adaptive moment estimation optimizer (Adam) are adopted with a learning rate of $0.001$ and a dropout loss rate of $0.2$. Our training dataset comprises $9000$ samples. The ratio of the training set size to the testing one is $4$. We use the small batch training scheme with a batch size of $50$ and set the number of training epochs between $50$ and $100$ depending on the convergence. 
The port-selection accuracy of the DL-JPS algorithm is calculated by averaging the correct port selection rate for all BS-user pairs over $N$ testing samples, i.e., 
\setcounter{equation}{41}
\begin{equation}
	\eta  = \frac{1}{{UB}}\sum\limits_{u = 1}^U {\sum\limits_{b = 1}^B  \left( \frac{1}{N}\sum\limits_{n = 1}^N {\frac{{\left| {\Lambda _{b,u}^{n,\star}} \right|}}{{\left| {{\Lambda ^n_{b,u}}} \right|}}}  \right)} \times  \text{100\%} ,
\end{equation}	
where $\Lambda _{b,u}^{n,\star} = \Lambda _{b,u}^n  \cap  \hat{\Lambda}^n_{{b,u}}$ with ${\Lambda _{b,u}^n}$ and $\hat{\Lambda}^n_{b,u}$ being the label port set and the output port set for the $n$-th sample.

	%\vspace{-0.5cm}
\subsection{Results and Discussions}	
	\begin{figure}[htbp!]
		%\vspace{-0.2cm}  % 调整与上文的间距
	%	\setlength{\abovecaptionskip}{0.1cm}   %调整图片标题与图距离
		%\setlength{\belowcaptionskip}{-1cm}   %调整图片标题与下文距离
		\centering 
		\includegraphics[width=1\linewidth]{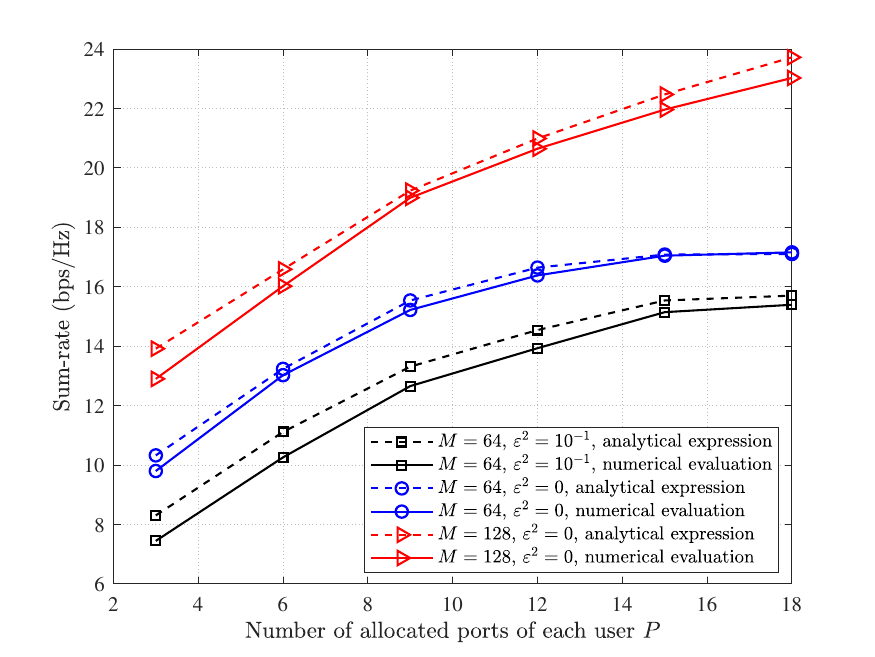}
		\caption{Analytical and numerical sum-rate versus the number of allocated ports of each user $P$. $L=20$,  ${L}_{0}=4$, $\text{SNR} = 15\,\text{dB}$, $\rho_{\rm{s}}=0$, $\rho_{\rm{c}}=1$.}
		\label{figure1}
	%	\vspace{-0.9cm}
	\end{figure}

	Fig. \ref{figure1} shows both the simulated and the analytical sum-rate results with the port selection resulted from performing Algorithm \ref{algorithm1} where $N_{\text{rand}}=100$ and $L_0 = 4$. 
	It can be seen that either with perfect estimation and feedback or with estimation and quantized feedback error, the analytical results, calculated from Eq. \eqref{equal_sumrate}, well match the simulated results in both cases with $M=64$ and $M=128$.  
	{To focus on the performance of the port selection feedback scheme, we omit the estimation and quantized feedback error of the port coefficients in the following.}
		
	\begin{figure}[htbp!]
		\vspace{-0.5cm}  % 调整与上文的间距
		\centering 
		\includegraphics[width=1\linewidth]{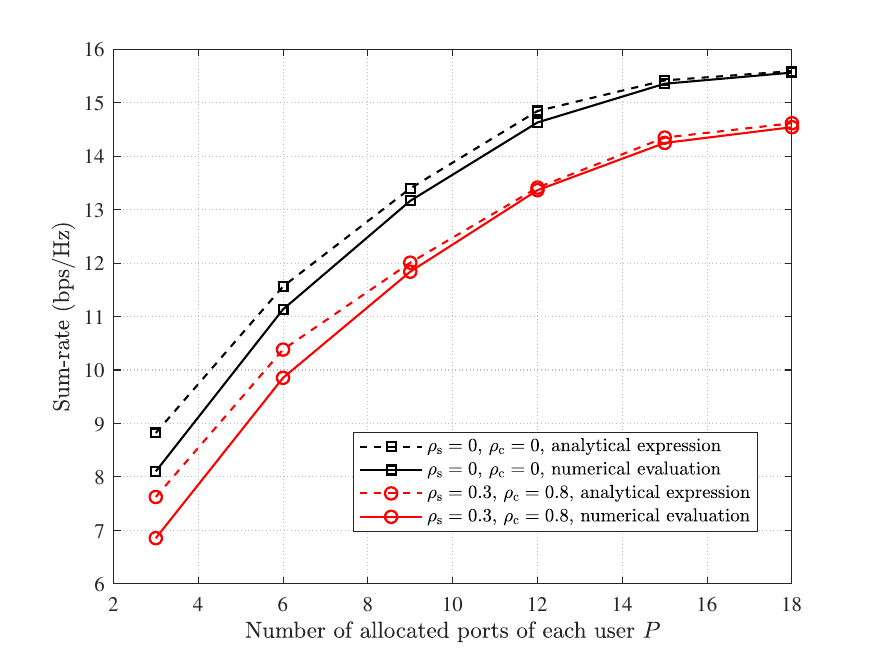}
		\caption{Sum-rate versus the number of allocated ports of each user $P$ with different port correlation levels. $M=64$, $L=20$,  ${L}_{0}=4$, $\text{SNR} = 15\,\text{dB}$.}
		\label{figure1_2}
		\vspace{-0.5cm}
	\end{figure}

	Fig. \ref{figure1_2} shows the sum-rate versus the number of allocated ports of each user $P$ with different beam-domain channel correlation levels for $M = 64$, $L = 20$, $L_0 = 4$, and $\text{SNR} = 15\,\text{dB}$.
	This proves that the derived analytical expressions are applicable to different port correlation levels, especially for $P \ge 10$, where the analytical results remain in high coincidence with the numerical results. In addition, the port correlation of multiple BSs causes some degradation in the performance of the MIMO architecture.
	Specifically, when $\rho_{\rm{s}}$ increases to $0.3$ for ports at the same BS and $\rho_{\rm{c}}$ increases to $0.8$ for ports at different BSs, the system sum-rate decreases by more than $1$ bps/Hz.

	  \begin{figure*}[htp]
	  	%\vspace{-0.7cm}  % 调整与上文的间距
	  	\setlength{\belowcaptionskip}{-5pt}
	  	\centering
	  	\subfigure[$\text{SNR} = 15\text{dB}$]{\label{fig6_a}
	  		\begin{minipage}[t]{0.5\linewidth}
	  			\centering
	  			\includegraphics[width=3.5in]{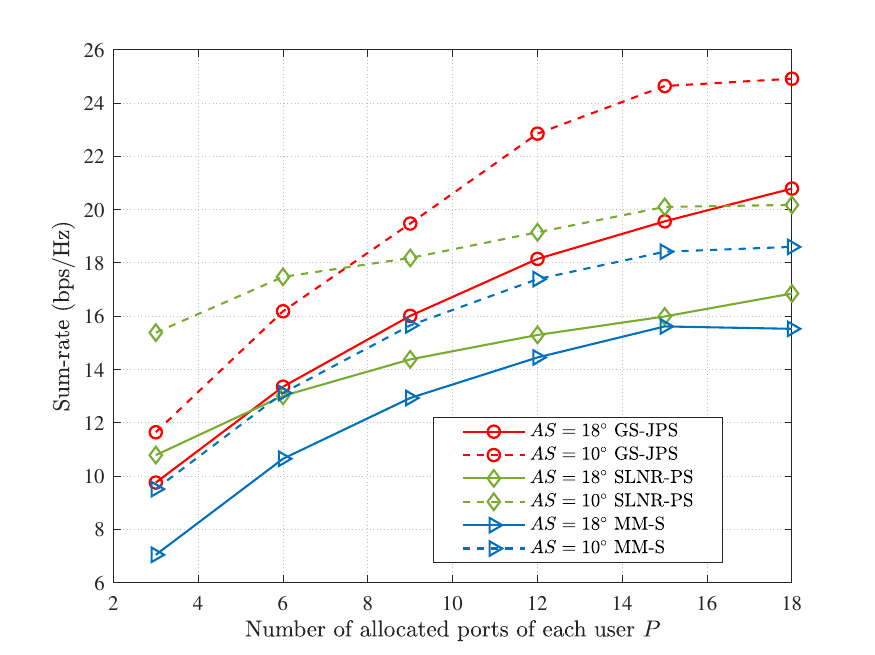}
	  			%\caption{fig1}
	  		\end{minipage}%
	  	}%
  	%\\
	  	\subfigure[$\text{SNR} = -10\text{dB}$]{\label{fig6_b}
	  		\begin{minipage}[t]{0.5\linewidth}
	  			\centering
	  			\includegraphics[width=3.5in]{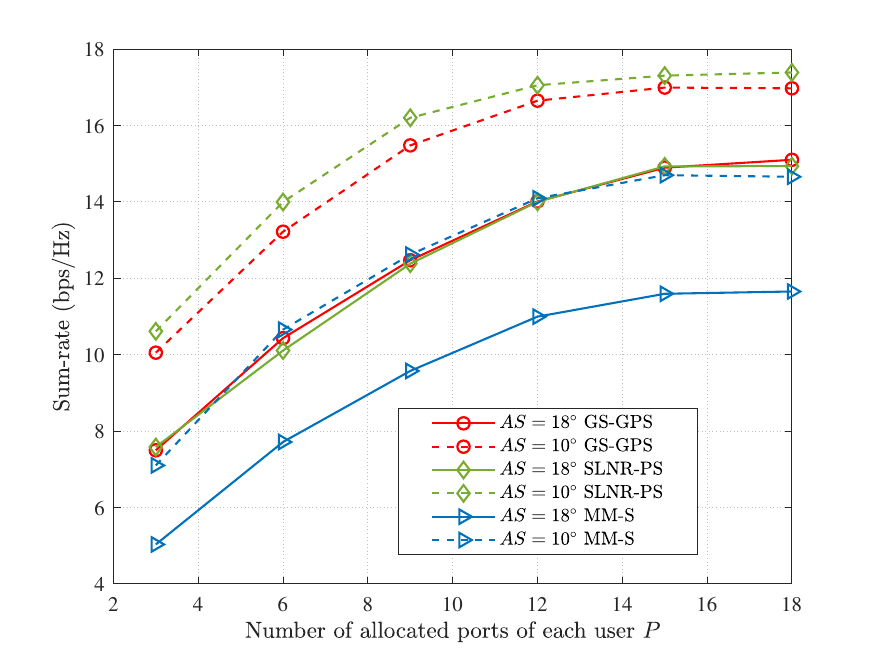}
	  			%\caption{fig2}
	  		\end{minipage}%
	  	}%
	  	\centering
	  	\caption{Sum-rate of the GS-JPS scheme, the SLNR-PS scheme, and the MM-S scheme versus the number of allocated ports of each user $P$. $M=64$, ${L}_{0}=4$, $\rho_{\rm{s}}=0$, $\rho_{\rm{c}}=1$.}
	  	\label{Figure2}
	  	%\vspace{-1.5cm}  % 调整与上文的间距
	  \end{figure*}
  
  Fig. \ref{Figure2} shows the sum-rate performance of the proposed GS-JPS algorithm, the SLNR-PS scheme, and the MM-S scheme for different numbers of allocated ports. Two cases where $L=20$ and $L=12$ are considered, and the corresponding angular spreads are about $AS=18{}^\circ $ and $AS=10{}^\circ $, respectively.
  In Fig. \ref{fig6_a} and Fig. \ref{fig6_b}, the SNR is set as $15\,\text{dB}$ and $-10\,\text{dB}$, respectively.
  It can be seen that larger $L$ causes performance degradation to all schemes due to more severe conflicts of effective ports among users.
  By comparing Fig. \ref{fig6_a} and Fig. \ref{fig6_b}, it can be concluded that the proposed GS-JPS scheme performs better in scenarios with large $P$ and high SNR, referred to as the interference-limited scenario.
  Specially, as shown in Fig. \ref{fig6_a}, when $\text{SNR} = 15\,\text{dB}$, the GS-JPS algorithm outperforms the SLNR-PS algorithm for $P\ge 9$ when 
  $AS=10{}^\circ $ and for $P\ge 6$ when $ AS=18{}^\circ$.
  For the case of $AS=10{}^\circ $ and $P=15$, the sum-rate of the GS-JPS algorithm is $22.6$\% higher than that of SLNR-PS and $34.0$\% higher than that of MM-S. 
  In the noise-limited scenario as shown in Fig. \ref{fig6_b}, i.e., $\text{SNR} = -10\,\text{dB}$, the GS-JPS and the SLNR-PS have similar sum-rate, but the GS-JPS needs much less computational complexity than the SLNR-PS.
  For $AS=10^\circ$ and $P=18$, the GS-JPS algorithm outperforms the MM-S algorithm by about $18.6$\% in sum-rate and is slightly better than the SLNR-PS. 
  %\textcolor{blue}{In addition, the complexity of the GS-JPS and SLNR-PS schemes is shown in Table \ref{complexity_compa}, and it can be found that the proposed scheme can reduce the number of multiplication than that of the SLNR-PS by one order.}

		The performance gain of the GS-JPS scheme over the SLNR-PS scheme in the interference-limited scenario comes from 1) the ZF precoding adopted in the GS-JPS scheme performs better in sum-rate compared to the SLNR precoding for high $\text{SNR}$; 2) the GS-JPS provides more efficient port selection compared to the subtractive port selection in the SLNR-PS scheme, especially under the constraint of no port sharing between users.

	Fig. \ref{figure3} shows the sum rate of the three port-selection schemes versus the number of different effective ports $L$, where $M=64$, $P=12$, ${L}_{0}=4$, and $\text{SNR} = 15$ dB. With increasing $L$, the average power distribution of the ports is more dispersed, and the sum-rates of the GS-JPS scheme, the SLNR-PS scheme, and the MM-S scheme all decrease due to the increased overlap of effective distinguishable ports among users. In addition, the GS-JPS algorithm outperforms the other two schemes consistently with a sum-rate advantage of $4$ bps/Hz. The reason for this performance gain has already been analyzed in the previous paragraph.
	
	\begin{figure}[htp!]
		\vspace{-0.5cm} 
		\centering 
		\includegraphics[width=1\linewidth]{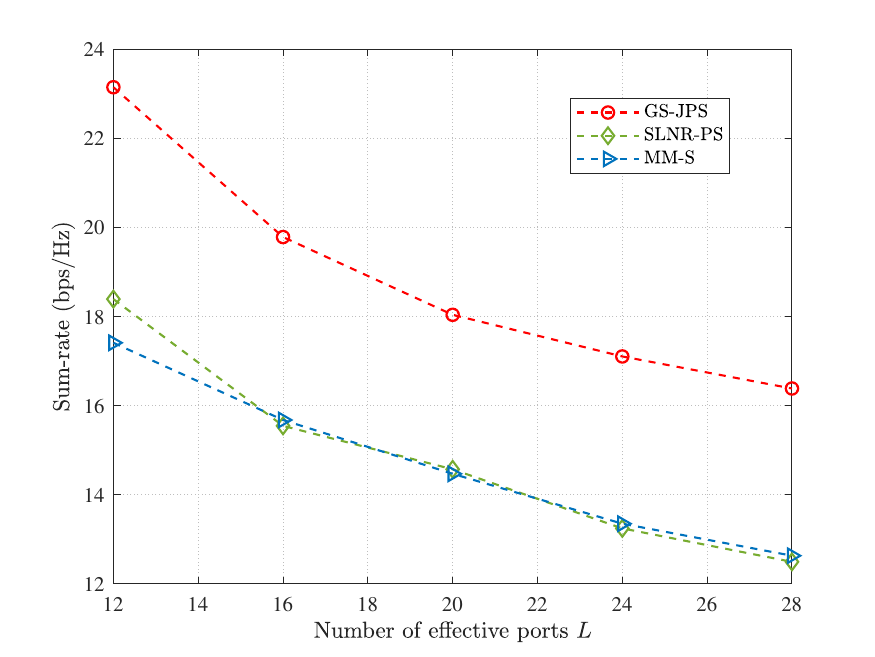}
		\caption{Sum-rate of the GS-JPS scheme, the SLNR-PS scheme and the MM-S scheme versus the number of effective ports $L$. $M=64$,  $P = 12$, ${L}_{0}=4$, $\text{SNR} = 15\,\text{dB}$, $\rho_{\rm{s}}=0$, $\rho_{\rm{c}}=1$.}
		\label{figure3}
		\vspace{-0.5cm}  % 调整与上文的间距
	\end{figure}

	\begin{figure*}[htp]
		\vspace{-0.6cm}
		\centering
		\subfigure[Sum-rate of different compression methods]{\label{fig4_a}
			
			\begin{minipage}[t]{0.5\linewidth}
				\centering
				\includegraphics[width=3.5in]{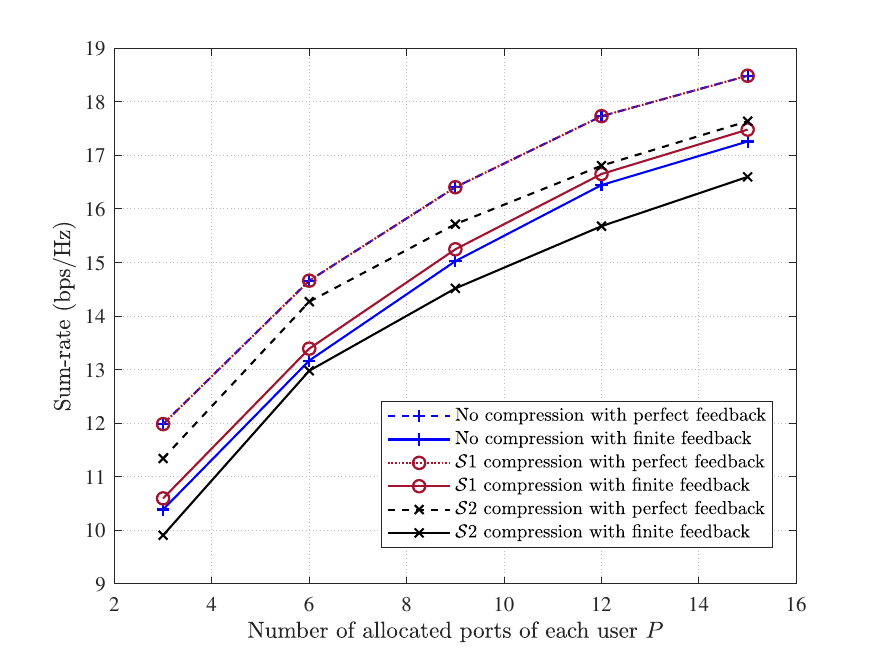}
				%\caption{fig1}
			\end{minipage}%
		}%
		%\\
		\subfigure[CCDF curves of the average compression ratio $\rm{CR_1}$ with $\mathcal{S}1$]{\label{fig4_b}
			\begin{minipage}[t]{0.5\linewidth}
				\centering
				\includegraphics[width=3.5in]{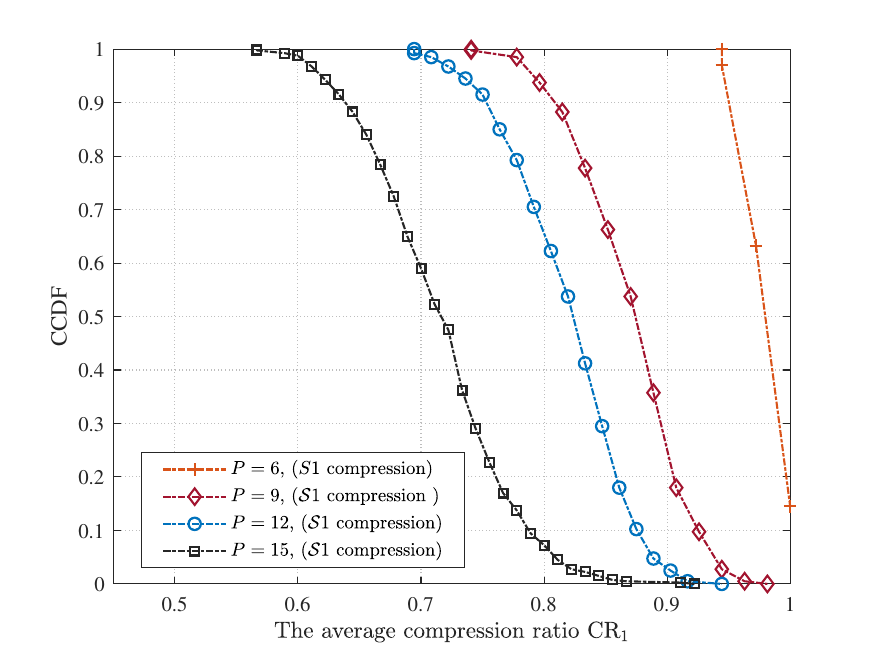}
				%\caption{fig2}
			\end{minipage}%
		}%
		\centering
		\caption{Performance of the EDT-based feedback algorithm. $M=64$, ${L}=20$, ${L}_{0}=12$, $\text{SNR} = 15\,\text{dB}$, $\rho_{\rm{s}}=0$, $\rho_{\rm{c}}=1$.}
		\label{Figure4}
		\vspace{-0.3cm}
	\end{figure*}

	Next, we study the performance of port-selection-based CSI feedback and reconstruction scheme. Some compression methods are first explained. 
	$\mathcal{S}1$ denotes the compression method according to Eq. \eqref{equal3_18}, and for the method $\mathcal{S}2$, the largest $\left \lceil \frac{3P}{4} \right \rceil $ eigenvalues of $\mathbf{R}_{{\Lambda}_{u}}$ and their corresponding eigenvectors are used to compress the port coefficients, i.e., ${{\mathbf{{r}}}_{u,\text{r}}}=\boldsymbol{\Sigma} _{u,\mathcal{S}2}^{-1/2}\mathbf{U}_{u,\mathcal{S}2}^{\rm{H}}{\bar{\mathbf{{h}}}_{\Lambda_u}}$, where ${\boldsymbol{\Sigma}_{u,\mathcal{S}2}}$ contains the largest $\left \lceil \frac{3P}{4} \right \rceil $ eigenvalues of $\mathbf{R}_{{\Lambda}_{u}}$ as its diagonal elements and ${{\mathbf{U}}_{u,\mathcal{S}2}}$ is composed of the corresponding eigenvectors of these eigenvalues.
	Further, scalar quantization is considered where $4$ bits and $3$ bits are used to represent the amplitude and the phase of each feedback coefficient, respectively.
	Denote ${\text{C}}_{u}=7P$, ${\text{C}}_{1,u}=7r_u$ and ${\text{C}}_{2,u}=7\left \lceil \frac{3P}{4} \right \rceil $ as the feedback overhead (bits) for no compression, compression using $\mathcal{S}1$ and $\mathcal{S}2$, respectively.
	 Define ${\rm{CR_1 = }}\sum\nolimits_{u = 1}^U {{\text{C}}_{1,u}} /\sum\nolimits_{u = 1}^U {{{\text{C}}_{u}}} $ as the average compression ratio for all users with $\mathcal{S}1$.

	Fig. \ref{fig4_a} plots the sum-rate under different compression strategies, with perfect and finite feedback denoting whether the CSI is quantized or not, respectively, and Fig. \ref{fig4_b} shows the CCDF curves for $\rm{CR_1}$ with 
	 $M=64$, ${L}=20$, ${L}_{0}=12$ and $\text{SNR} = 15\,\text{dB}$.
	Fig. \ref{fig4_a} shows that $\mathcal{S}2$ with perfect feedback reduces the overhead by about $25$\% while only losing $5$\% sum-rate. 
 $\mathcal{S}1$ with perfect feedback, on the other hand, has negligible sum-rate performance loss, and the overhead is reduced by over $25$\% with a probability of $77$\% for $P =15$.
	Notice that 
	$\mathcal{S}1$ also performs better than the no-compression case in terms of sum-rate when finite feedback is considered. This is because the eigenvalues of port coefficient space are more suitable for scalar quantization due to their much lower correlation than that between port coefficients.
	%\textcolor{blue}{Regardless of whether the SNR is high or low, the CSI recovered under the	$\mathcal{S}1$ compression scheme is more accurate than that with no compression feedback, and thus the performance of the sum-rate will be greater.}
	In Fig. \ref{fig4_b}, the compression rate $\rm{CR_1}$ decreases with the selection of more ports, because more ports with correlation make the ratio of the rank to $P$ smaller.
	
	\begin{figure}[htbp!]
		%\vspace{-0.3cm}
		%\setlength{\abovecaptionskip}{0.1cm}   %调整图片标题与图距离
		\centering 
		\includegraphics[width=1\linewidth]{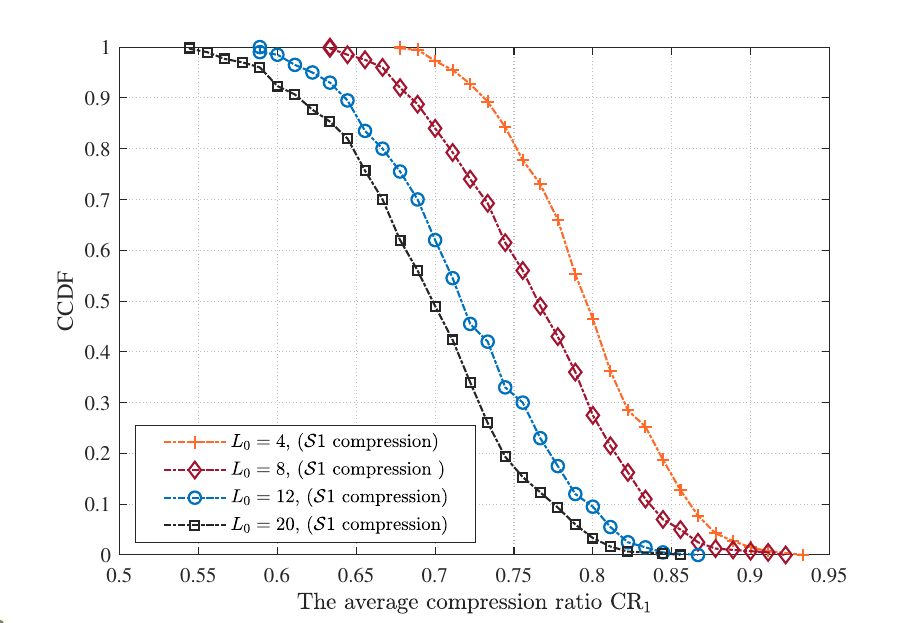}
		\caption{CCDF curves of the average compression ratio $\rm{CR_1}$  with $\mathcal{S}1$ compression method versus the number of correlated ports $L_0$. $M=64$, ${L}=20$, $P = 15$, $\text{SNR} = 15\,\text{dB}$, $\rho_{\rm{s}}=0$, $\rho_{\rm{c}}=1$.}
		\label{figure5}
		%\vspace{-0.2cm}
	\end{figure}
	
	Fig. \ref{figure5} plots CCDF curves of the compression ratio $\rm{CR_1}$ versus the number of correlated ports $L_0$ between multiple BSs at each user with $M=64$, ${L}=20$, $P = 15$ and $\text{SNR} = 15\,\text{dB}$.
	With increasing $L_0$ the compression margin using the $\mathcal{S}1$ method becomes larger, i.e., $\mathbf{R}_{{\Lambda}_{u}}$ has a small rank with greater probability. 
	Therefore, the proposed EDT-based feedback algorithm, especially the $\mathcal{S}1$ method, has desirable performance in scenarios with abundant common scatterers.
	
\begin{figure*}[htbp!]
	%\vspace{-0.5cm}
	\centering
	\subfigure[Accuracy of port prediction]{\label{figure6}
		
		\begin{minipage}[t]{0.5\linewidth}
			\centering
			\includegraphics[width=3.5in]{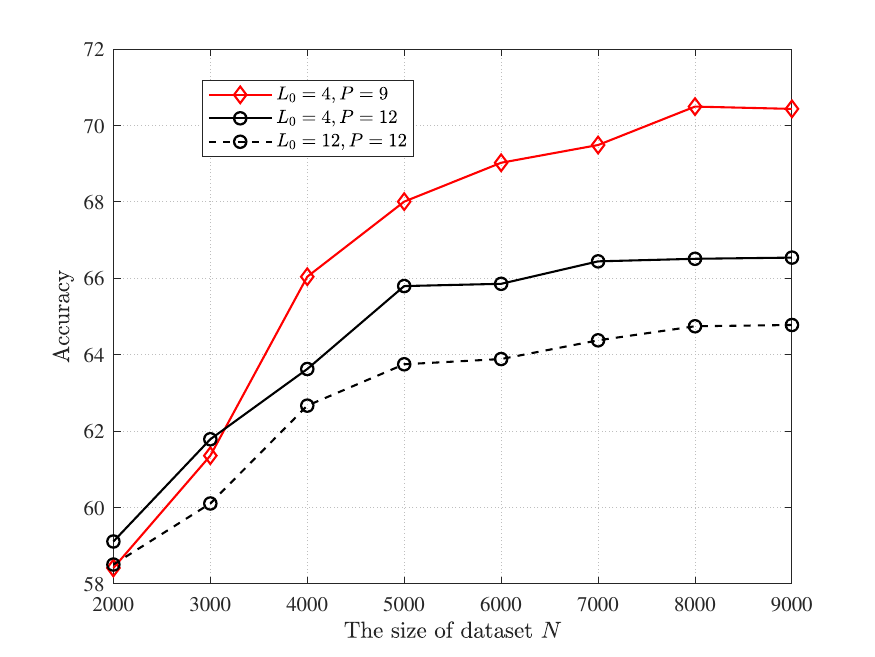}
			%\caption{fig1}
		\end{minipage}%
	}%
	%\\
	\subfigure[Sum-rate based on the proposed algorithms]{\label{figure7}
		\begin{minipage}[t]{0.5\linewidth}
			\centering
			\includegraphics[width=3.5in]{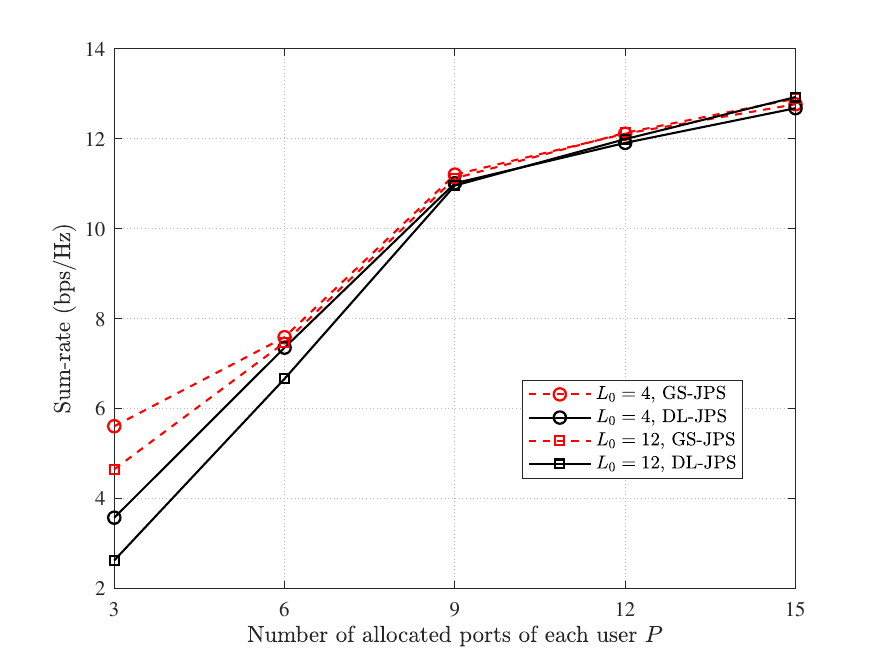}
			%\caption{fig2}
		\end{minipage}%
	}%
	\centering
	\caption{Performance based on the port selection of the DL-JPS algorithm. $M=32$, ${L}=20$, $\text{SNR} = 15\,\text{dB}$.}
	\label{Figure4}
	%\vspace{-0.5cm}
\end{figure*}

   Fig. \ref{figure6} shows the accuracy of DL-JPS for varying dataset size with $M=32$, ${L}=20$, $\text{SNR} = 15\,\text{dB}$. It can be concluded that the accuracy of port prediction initially increases rapidly as the sample size increases and eventually reaches a plateau. 
When the sample size is $6000$ and $9000$, the accuracy of port prediction for $P=12$ and $L_{0} = 4$ is about 66.0\% and 66.5\%, respectively, while the accuracy for $P=9$ and $L_{0} = 4$ is around 69.1\% and 71.2\%, respectively.
	As for $L_0=12$, due to the increased correlation of the port coefficients, the input-output relationship of the GS-JPS algorithm becomes more complex, resulting in relatively low accuracy of the DL-JPS learning. Specifically, the prediction accuracy of the DL-JPS scheme decreases by about $2$\% from $L_{0} = 4$ to $L_{0} = 12$ with the same $P=12$.
 
	Fig. \ref{figure7} shows the sum-rate based on the proposed GS-JPS scheme and DL-JPS scheme when $L_0 = 4 $ and $L_0 = 12 $ with data size $N=9000$. It can be seen that the DL-JPS scheme achieves a comparable sum-rate to that of the GS-JPS scheme. 
	Considering that the online execution complexity of DL-JPS comes from only one forward propagation of the employed deep network, its applicability in fast time-varying scenarios can be effectively guaranteed.

	\section{Conclusion}\label{Conclusion}
	\vspace{0em}
	In this paper, we study the FDD cell-free massive MIMO downlink with zero-forcing precoding and under a general spatial domain channel model with port coefficient correlation and heterogeneous average port power profile. A joint-port-selection-based channel acquisition and feedback scheme was proposed which uses an EDT-based algorithm to reduce the feedback overhead by sufficiently exploring the port correlation.  
	Further, we derived an expression of the system sum-rate as a function of the port coefficient correlation, average port power, and other system parameters.
	We then formulated the port selection problem to maximize the sum-rate.
	%\textcolor{blue}{We exploited both the partial reciprocity between uplink and downlink channels and the port correlation resulting from common scatterers to guide the port selection.} 
	As the size of the search space for the port selection grows exponentially with the total number of antennas and the total number of selected ports, we proposed two low-complexity schemes, the GS-JPS algorithm and its DL-aided imitator algorithm DL-JPS. 
	Simulations show that our derived sum-rate expression is accurate and the overall proposed scheme consisting of the EDT feedback and the GS-JPS or the DL-JPS performs better in sum-rate with comparable feedback overhead, compared to baseline schemes, e.g., the SLNR-PS and MM-S schemes.
	%\textcolor{blue}{In the future, for the proposed port-selective channel acquisition scheme, the impact on the sum-rate performance deserves to be investigated in depth under a more generalized model of port coefficient estimation and quantitative feedback.} 
	\vspace{-0em}

	%\section{}
	%Lorem ipsum dolor sit amet, consectetur adipisicing elit, sed do

	%\section*{Acknowledgment}

	\section{Appendix}

	\subsection{Proof of Lemma \ref{lemma_1}}\label{{Proof of Lemma 1}}
	
	%Define $\zeta_u \triangleq \sum\limits_{b=1}^{B}{\sum\limits_{k \in {\Lambda _{b,u}}} {{{\bar \beta }_{b,u,k}}{{\left| {{{\bar h}_{b,u,k}}} \right|}^2}}} $, and ${{\bf{B }}_{\Lambda_u}} = {\rm{blk}}[{{\bf{B }}_{\Lambda_{1,u}}},\ldots,{{\bf{B }}_{\Lambda_{B,u}}}]$, we have 
	\textcolor{black}{
	Define $\zeta_u \triangleq \sum\limits_{b=1}^{B}{\sum\limits_{k \in {\Lambda _{b,u}}} {{{\bar \beta }_{b,u,k}}{{\left| {{\hat{\bar h}_{b,u,k}}} \right|}^2}}} $, we have 
	\begin{equation}\label{zc_eq_14}
		\begin{aligned}
			\zeta_u &= {\hat{\bar{\bf{h}}}}^{\rm{H}}_{{\Lambda}_{u}}
			{{\bf{B }}_{\Lambda_u}}{{\bf{B }}^{\rm{H}}_{\Lambda_u}}
			{\hat{\bar{\bf{h}}}}_{{\Lambda}_{u}},\\
			& = {\hat{\bar{\bf{h}}}}^{\rm{H}}_{{\Lambda}_{u}}
			{ {\hat{\bf{R}}_{\Lambda_u}}} ^{-{\rm{1}}/{2}}{{\bf{S }}_{u}}{ {\hat{\bf{R}}_{\Lambda_u}}} ^{-{\rm{1}}/{2}}
			{\hat{\bar{\bf{h}}}}_{{\Lambda}_{u}}.
		\end{aligned}
	\end{equation}
Given the EVD of the positive semi-definite ${{\bf{S}}_u}$: ${{\bf{S}}_u} = {{\bf{U}}_{u,{\rm{s}}}}{{\bf{\Sigma }}_{u,{\rm{s}}}}{\bf{U}}_{u,{\rm{s}}}^{\rm{H}}$, $\zeta_u$ can be further expressed as
	\begin{equation} \label{zetau_appendix}
		\begin{aligned}
			{\zeta _u} 
			 &= {\hat{\bar{\bf{h}}}}^{\rm{H}}_{{\Lambda}_{u}}
			{ {\hat{\bf{R}}_{\Lambda_u}}} ^{-{\rm{1}}/{2}}
			{{\bf{U}}_{u,{\rm{s}}}}{{\bf{\Sigma }}_{u,{\rm{s}}}}{\bf{U}}_{u,{\rm{s}}}^{\rm{H}}
		{ {\hat{\bf{R}}_{\Lambda_u}}} ^{-{\rm{1}}/{2}}
			{\hat{\bar{\bf{h}}}}_{{\Lambda}_{u}},\\
			 &= {\check{\bf{h}}^{\rm{H}}_{{\Lambda}_{u}}}{{\boldsymbol{\Sigma }}_{u,{\rm{s}}}}{\check{\bf{h}}_{{\Lambda}_{u}}} ,\\
			&=\sum\limits_{i = 1}^{{\rho _u}} {\frac{{{\lambda _{u,i}}}}{2}{{\left| {\sqrt 2 \check h_{\Lambda_u,i}} \right|}^2}}  ,
		\end{aligned}
	\end{equation}
	where ${\check{\bf{h}}_{{\Lambda}_{u}}} \triangleq {\bf{U}}_{u,{\rm{s}}}^{\rm{H}}
	{ {\hat{\bf{R}}_{\Lambda_u}}} ^{-{\rm{1}}/{2}}
	{\hat{\bar{\bf{h}}}}_{{\Lambda}_{u}}$ and ${\check h_{\Lambda_u,i}}$ is the $i$-th element of ${\check{\bf{h}}_{{\Lambda}_{u}}}$. Since  ${{\hat{\bar{\bf{h}}}}_{{\Lambda}_{u}}} \sim \mathbb{CN}({\bf{0}},{ {\hat{\bf{R}}_{\Lambda_u}}})$, it can be shown that ${\check{\bf{h}}_{{\Lambda}_{u}}} \sim \mathbb{CN}({\bf{0}},{\bf{I}}_{K_u})$ and ${\left| {\sqrt 2 \check h_{\Lambda_u,i}} \right|^2} \sim \chi _2^2\left( 0 \right)$.}
	From \cite[Eq. (2.4)]{provost1996exact}, the PDF of $\zeta_u$ can be written as 
	\begin{equation}
		\label{PDF_zeta_u}
		{f_{{\zeta _u}}}\left( x \right) = \sum\limits_{k = 0}^\infty  {\frac{{{\alpha_{u,k}}}}{{\Gamma \left( {{\rho _u} + k} \right){{\left( {2{\beta _u}} \right)}^{{\rho _u} + k}}}}} {x^{{\rho _u} + k - 1}}{e^{ - \frac{1}{{2{\beta _u}}}x}},
	\end{equation}
	where the parameters are defined in Eq. \eqref{para_lemma1}. Therefore,
	\begin{equation}
		\label{expectation}
		\begin{aligned}
			&\,\quad\mathbb{E}\left\{ {{\left\| {{{{\bar{\mathbf w}}}}_{u}} \right\|}^{2}} \right\}=\frac{1}{M}\mathbb{E}\left\{ {\frac{1}{{{\zeta _u}}}} \right\} ,\\
			&= \frac{1}{M}\sum\limits_{k = 0}^\infty  {\frac{{{\alpha_{u,k}}}}{{\Gamma \left( {{\rho _u} + k} \right){{\left( {2{\beta _u}} \right)}^{{\rho _u} + k}}}}} \int_0^\infty  {{x^{{\rho _u} + k - 2}}} {e^{ - \frac{1}{{2{\beta _u}}}x}}\mathrm{d}x,\\
			&= \frac{1}{M}\sum\limits_{k = 0}^\infty  {\frac{{{\alpha_{u,k}}}{\left( {2{\beta _u}} \right)^{{\rho _u} + k - 1}}}{{\Gamma \left( {{\rho _u} + k} \right){{\left( {2{\beta _u}} \right)}^{{\rho _u} + k}}}}}  \Gamma \left( {{\rho _u} + k - 1} \right),\\
			&= \frac{1}{M}\sum\limits_{k = 0}^\infty  {\frac{{{\alpha_{u,k}}}}{{2{\beta _u}\left( {{\rho _u} + k - 1} \right)}}} , \,{\text{for}}\,\,{\rho _u} > 1.
		\end{aligned}
	\end{equation}

	%According to Eq. \eqref{expectation}, Eq. \eqref{equal_lemma1} can be obtained.

	\subsection{Proof of Lemma \ref{lemma2}}
	
	\newcounter{TempEqCnt4}                         % 创建临时变量TempEqCnt
	\setcounter{TempEqCnt4}{\value{equation}} % 将当前公式序号 赋给TempEqCnt
	\setcounter{equation}{48}  % 当前公式序号变为x，x等于长公式应有的序号减1.
	\begin{figure*}[hb] % hb底部，ht为头部
		\centering % 公式居中
		\hrulefill % 添加一条水平线
		\textcolor{black}{
			\begin{equation} \label{eq_44} 
				\begin{aligned}
					{{\left[ {{{\hat{\mathbf{H}}}}^{\rm{H}}} {{{\tilde{\mathbf{h}}}}_{u}}\tilde{\mathbf{h}}_{u}^{\rm{H}} \hat{\mathbf{H}} \right]}_{{{v}_{1}},{{v}_{2}}}}&=
					{M^2}\sum\limits_{b' = 1}^B {\sum\limits_{b = 1}^B {\hat {\bar {\bf{h}} }}_{b,{v_1}}^{\rm{H}}{\bf{B}}_{{\Lambda _{b,{v_1}}}}^{\rm{H}}{{\bf{B}}_{\Lambda _{b,u}^C \cap {\Lambda _{b,{v_1}}}}}{{{\bf{\bar h}}}_{b,u}}{\bf{\bar h}}_{b',u}^{\rm{H}}{\bf{B}}_{\Lambda _{b',u}^C \cap {\Lambda _{b',{v_2}}}}^{\rm{H}}{{\bf{B}}_{{\Lambda _{b',{v_2}}}}}{\hat {\bar {\bf{h}} }}_{b',{v_2}}} \\
					&\quad+ {M^2}\sum\limits_{b' = 1}^B {\sum\limits_{b = 1}^B {\hat {\bar {\bf{h}}} }_{b,{v_1}}^{\rm{H}}{\bf{B}}_{{\Lambda _{b,{v_1}}}}^{\rm{H}}{\bf{F}}_{\Lambda _{b,{v_1}}^{}}^{\rm{H}}{{\bf{F}}_{\Lambda _{b,u}^{}}}{{\bf{B}}_{\Lambda _{b,u}^{}}}{\tilde{\bar {\bf{ h}}}}_{b,u}}{\tilde{\bar {\bf {h}}}}_{b',u}^{\rm{H}}{\bf{B}}_{\Lambda _{b',u}^{}}^{\rm{H}}{\bf{F}}_{\Lambda _{b',u}^{}}^{\rm{H}}{{\bf{F}}_{\Lambda _{b',{v_2}}^{}}}{{\bf{B}}_{{\Lambda _{b',{v_2}}}}}{\hat {\bar {\bf{h}}} }_{b',{v_2}} .
				\end{aligned}
		\end{equation} }
		%\hrulefill % 添加一条水平线
	\end{figure*}
	\setcounter{equation}{\value{TempEqCnt4}} % 把TempEqCnt中存的公式序号赋回给当前公式序号

	\begin{figure*}[hb] % hb底部，ht为头部
		\centering % 公式居中
		%\hrulefill % 添加一条水平线
		\vspace*{-20pt} % 调整线与公式之间的距离
		\newcounter{TempEqCnt5}                         % 创建临时变量TempEqCnt
		\setcounter{TempEqCnt5}{\value{equation}} % 将当前公式序号 赋给TempEqCnt
		\setcounter{equation}{50}  % 当前公式序号变为x，x等于长公式应有的序号减1.
		\textcolor{black}{
			\begin{equation}
				\label{xi_uv_up}
				{\xi_{u,v}^{\text{up}}} =
				\begin{cases}
					{M^2}\sum\limits_{b' = 1}^B {\sum\limits_{b = 1}^B {\hat {\bar {\bf{h}}} }_{b,u}^{\rm{H}}{\bf{B}}_{{\Lambda _{b,u}}}^{\rm{H}}{{\bf{B}}_{\Lambda _{b,u}^{}}}{\tilde{\bar {\bf { h}}}}_{b,u}{\tilde{\bar {\bf {h}}}}_{b',u}^{\rm{H}}{\bf{B}}_{\Lambda _{b',u}^{}}^{\rm{H}}{{\bf{B}}_{{\Lambda _{b',u}}}}{\hat {\bar {\bf{h}}} }_{b',u}} ,& {\text{ if }}\, v=u\\
					{M^2}{\sum\limits_{b'=1}^{B}{\sum\limits_{b=1}^{B}{\hat{\bar{\bf{h}}}}_{b,v}^{\rm{H}}\mathbf{B}_{{{\Lambda }_{b,v}}}^{\rm{H}}{{\mathbf{B}}_{\Lambda _{b,u}^{C}\cap {{\Lambda }_{b,v}}}} {{{\mathbf{\bar{h}}}}_{b,u}}\mathbf{\bar{h}}_{b',u}^{\rm{H}} \mathbf{B}_{\Lambda _{b',u}^{C}\cap {{\Lambda }_{b',v}}}^{\rm{H}}{{\mathbf{B}}_{{{\Lambda }_{b',v}}}}{{\hat{\bar{\bf{h}}}}_{b',v}}}},
					& {\text{ if }}\, v \ne u.
				\end{cases} 
		\end{equation} }
	\setcounter{equation}{\value{TempEqCnt5}} % 把TempEqCnt中存的公式序号赋回给当前公式序号
		%\vspace*{8pt} % 调整线与公式之间的距离
		%	\hrulefill % 添加一条水平线
		\newcounter{TempEqCnt6}                         % 创建临时变量TempEqCnt
		\setcounter{TempEqCnt6}{\value{equation}} % 将当前公式序号 赋给TempEqCnt
		\setcounter{equation}{55}  % 当前公式序号变为x，x等于长公式应有的序号减1.
		\textcolor{black}{
			\begin{equation} \label{xi_uv_up_y}
				\xi_{u,v}^{{\rm{up}},y} = 
				\begin{cases}
					M^2\sum\limits_{b' = 1}^B {\sum\limits_{b \ne b'}^B {\sum\limits_{l \in {\Lambda _{b,u}}} {\sum\limits_{l' \in {\Lambda _{b',u}}} {
									{{{\bar \beta }_{b,u,l}}{{\bar \beta }_{b',u,l'}}{{\tilde {\bar h}}_{b,u,l}}\tilde {\bar h}_{b',u,l'}^*{\hat {\bar h}_{b,u,l}^*}{{\hat {\bar h}}_{b',u,l'}}} } } } } ,
					& {\text{ if }}\, v=u \\
					M^2\sum\limits_{b' = 1}^B {\sum\limits_{b  \ne  b'}^B {\sum\limits_{l \in {\Lambda _{b,v}}} {\sum\limits_{l' \in {\Lambda _{b',v}}} {\sqrt {{{\bar \beta }_{b,v,l}}{{\bar \beta }_{b,u,l}}{{\bar \beta }_{b',v,l'}}{{\bar \beta }_{b',u,l'}}}{{\bar h}_{b,u,l}}{{\bar h}^*_{b',u,l'}}
									{\hat{\bar {h}}_{b,v,l}^*{{\hat{\bar{ h}}}_{b',v,l'}} }  } } } }  ,		
					& {\text{ if }}\, v \ne u.
				\end{cases}
		\end{equation} }
		\setcounter{equation}{\value{TempEqCnt6}} % 把TempEqCnt中存的公式序号赋回给当前公式序号
	\end{figure*}

	From Eq. \eqref{eqn_ZF}, we obtain
	\begin{equation} \label{equal_lemma2_1}
		\begin{aligned}
			|{{{\tilde{\mathbf{h}}}}^{\rm{H}}_{u}}\mathbf{\bar{w}}_{v}|^2 
			&= {{\left[ {{\left( {{{\hat{\mathbf{H}}}}^{\rm{H}}}\hat{\mathbf{H}} \right)}^{-1}}{{{\hat{\mathbf{H}}}}^{\rm{H}}} {{{\tilde{\mathbf{h}}}}_{u}}\tilde{\mathbf{h}}_{u}^{\rm{H}} \hat{\mathbf{H}}{{\left( {{{\hat{\mathbf{H}}}}^{\rm{H}}}\hat{\mathbf{H}} \right)}^{-1}} \right]}_{v,v}} 
		\end{aligned}.
	\end{equation}
	For $\mathbf{i}=M\left( i-1 \right)+1:Mi$ and $\mathbf{k}=M\left( k-1 \right)+1:Mk$, ${{\left[  {{{\tilde{\mathbf{h}}}}_{u}}\tilde{\mathbf{h}}_{u}^{\rm{H}}  \right]}_{\mathbf{i},\mathbf{k}}}\in {{\mathbb{C}}^{M\times M}}$ is the $(i,k)$-th $M \times M$ block of ${{ {{{\tilde{\mathbf{h}}}}_{u}}\tilde{\mathbf{h}}_{u}^{\rm{H}} }}$. It can be calculated as
	\textcolor{black}{
	\begin{equation} 
		\begin{aligned}
			{{\left[  {{{\tilde{\mathbf{h}}}}_{u}}\tilde{\mathbf{h}}_{u}^{\rm{H}}  \right]}_{\mathbf{i},\mathbf{k}}} &=  
			\tilde{\mathbf{h}}_{i, u}\tilde{\mathbf{h}}^{\rm{H}}_{k, u},\\
			&=
			M{{\mathbf{F}}_{\Lambda _{i,u}^{C}}}{{\mathbf{B}}_{\Lambda _{i,u}^{C}}} {{{\mathbf{\bar{h}}}}_{i,u}}\mathbf{\bar{h}}_{k,u}^{\rm{H}} \mathbf{B}_{\Lambda _{k,u}^{C}}^{\rm{H}}\mathbf{F}_{\Lambda _{k,u}^{C}}^{\rm{H}} \\
			&\quad+M{{\bf{F}}_{\Lambda _{i,u}}}{{\bf{B}}_{\Lambda _{i,u}}}{{{\tilde{\bar {\bf h}}}}_{i,u}}{\tilde{\bar {\bf h}}}_{k,u}^{\rm{H}}{\bf{B}}_{\Lambda _{k,u}^{}}^{\rm{H}}{\bf{F}}_{\Lambda _{k,u}^{}}^{\rm{H}}\\
			&\quad + M{{\bf{F}}_{\Lambda _{i,u}^C}}{{\bf{B}}_{\Lambda _{i,u}^C}}{{{\bf{\bar h}}}_{i,u}}{\tilde{\bar {\bf h}}}_{k,u}^{\rm{H}}{\bf{B}}_{\Lambda _{k,u}^{}}^{\rm{H}}{\bf{F}}_{\Lambda _{k,u}^{}}^{\rm{H}}\\
			&\quad + M{{\bf{F}}_{\Lambda _{i,u}^{}}}{{\bf{B}}_{\Lambda _{i,u}^{}}}{{{\tilde{\bar {\bf h}}}}_{i,u}}{\bf{\bar h}}_{k,u}^{\rm{H}}{\bf{B}}_{\Lambda _{k,u}^C}^{\rm{H}}{\bf{F}}_{\Lambda _{k,u}^C}^{\rm{H}}.
		\end{aligned}
	\end{equation}}Recall that $\Lambda _{b,u}^{C}\cap {{\Lambda }_{b,u}}=\emptyset  ,\forall u \in \mathbb{U}$, ${{\left[ {{{\hat{\mathbf{H}}}}^{\rm{H}}} {{{\tilde{\mathbf{h}}}}_{u}}\tilde{\mathbf{h}}_{u}^{\rm{H}} \hat{\mathbf{H}} \right]}_{{{v}_{1}},{{v}_{2}}}}$, ${{v}_{1}},{{v}_{2}}\in \mathbb{U}$ can be represented as Eq. \eqref{eq_44} at the bottom of the next page, where ${{\mathbf{B}}_{\Lambda _{b,u}^{C}\cap {{\Lambda }_{b,{{v}_{1}}}}}}$ consists of the rows of ${{\mathbf{B}}_{b,u}}$ with indices $i \in {{\Lambda }_{b,{{v}_{1}}}}$ if $v_1 \ne u$.
	If $v_1 = u$, then ${{\mathbf{B}}_{\Lambda _{b,u}^{C}\cap {{\Lambda }_{b,{{v}_{1}}}}}} = {\bf{0}}$.
	From Eq. \eqref{equal3_4}, Eq. \eqref{equal_lemma2_1} can be calculated as
	 \setcounter{equation}{49}
	\begin{equation}\label{equal3_6}
		|{{{\tilde{\mathbf{h}}}}^{\rm{H}}_{u}}\mathbf{\bar{w}}_{v}|^2  =\frac{{\xi_{u,v}^{\text{up}}}}{\xi _{v}^{{\rm{down}}}},
	\end{equation}
	where  ${\xi_{u,v}^{\text{up}}}$ is defined as Eq. \eqref{xi_uv_up} at the bottom of the next page, which can be briefly derived from Eq. \eqref{eq_44} , $\xi _{v}^{{\rm{down}}} =M^2{\zeta^2 _{v}}$, and ${\zeta _{v}} = {\sum\limits_{b = 1}^B {\sum\limits_{k \in {\Lambda _{b,v}}} {{{\bar \beta }_{b,v,k}}{{\left| {{\hat{\bar h}}_{b,v,k}} \right|}^2}} } }$.

	First, we analyze the expression for the numerator with the case of $v=u$, the second half of ${\xi_{u,v}^{\text{up}}}$ satisfies
	\setcounter{equation}{51}
	\begin{equation}
		\begin{aligned}
			&\quad	{ \left[ {\tilde{\bar {\bf{ h}}}}_{b,u} {\tilde{\bar {\bf {h}}}}_{b',u}^{\rm{H}}{\bf{B}}_{\Lambda _{b',u}}^{\rm{H}}{\bf{B}}_{\Lambda _{b',u}}{\hat {\bar {\bf{h}}}}_{b',u}\right]_m} \\
			&= \sum\limits_{l' \in {\Lambda _{b',u}}} {{{\bar \beta }_{b',u,l'}}} {\tilde {\bar h}_{b,u,m}}\tilde {\bar h}_{b',u,l'}^*{\hat {\bar h}_{b',u,l'}},\forall m = 1,...,M.
		\end{aligned}
	\end{equation}
	And ${\xi_{u,v}^{\text{up}}}$ can be rewritten as
	\begin{equation}
		\begin{aligned}
		&\xi _{u,v}^{{\rm{up}}} = \\
			&M^2\sum\limits_{b' = 1}^B 
			\sum\limits_{b = 1}^B 
			\sum\limits_{l \in {\Lambda _{b,u}}}^{}\! 
			\sum\limits_{l' \in {\Lambda _{b',u}}}^{}\!\!\! 
			{{\bar \beta }_{b,u,l}}{{\bar \beta }_{b',u,l'}}
			{{\tilde {\bar h}}_{b,u,l}}\tilde {\bar h}_{b',u,l'}^*
			{\hat {\bar h}_{b,u,l}^*{{\hat {\bar h}}_{b',u,l'}}}.
		\end{aligned}
	\end{equation}

	In the same way,  ${\xi_{u,v}^{\text{up}}}$ in the case of $v \ne u$ can be expressed as 
	\begin{equation}
		\begin{aligned}
			{\xi_{u,v}^{\text{up}}}
			 = M^2\sum\limits_{b' = 1}^B \sum\limits_{b = 1}^B \sum\limits_{l \in {\Lambda _{b,v}}} \sum\limits_{l' \in {\Lambda _{b',v}}}\!\!\!   
			&\sqrt {{{\bar \beta }_{b,v,l}}{{\bar \beta }_{b,u,l}}{{\bar \beta }_{b',v,l'}}{{\bar \beta }_{b',u,l'}}} \\
			&\times {{\bar h}_{b,u,l}}{{\bar h}^*_{b',u,l'}}
			{\hat{\bar h}}_{b,v,l}^*{\hat{\bar h}_{b',v,l'}}  .
		\end{aligned}
	\end{equation}

	 Due to the correlation between ${\xi_{u,v}^{\text{up}}}$ and $\zeta_v$, both of which follow the generalized chi-squared distribution, it is challenging to compute $\mathbb{E}\left\{ \frac{{\xi_{u,v}^{\text{up}}}}{{\xi _v^{{\rm{down}}}}} \right\}$. 
	 Note that ${\xi_{u,v}^{\text{up}}}$ can be further split into $\xi_{u,v}^{{\rm{up}},x}$ $+$ $\xi_{u,v}^{{\rm{up}},y}$, where
	 \begin{equation}
	 	\linespread{1.0} \selectfont
	 	\xi_{u,v}^{{\rm{up}},x} =
	 	\begin{cases}
	 	 M^2\sum\limits_{b = 1}^B 
	 		\left|\sum\limits_{k \in {\Lambda _{b,u}}}^{} 
	 			{{\bar \beta }_{b,u,k}}{{\tilde {\bar h}}_{b,u,k}}{\hat {\bar h}_{b,u,k}^*}
	 		\right|^2,
	 		& {\text{ if }}\, v=u \\
	 		M^2\sum\limits_{b = 1}^B 
	 		\left|
	 		{\sum\limits_{k \in {\Lambda _{b,v}}}\!\! \!
	 			{\sqrt{{{\bar \beta }_{b,v,k}}{{\bar \beta }_{b,u,k}}}
	 				{{ {{{\bar h}_{b,u,k}}} }}	
	 				{{ {{\hat{\bar h}^*_{b,v,k}}} }}} }
 			\right|^2\!\!\!\!\!\!\!& {,\text{if }}\, v\ne u,
	 	\end{cases}
	 \end{equation}
	 and $\xi_{u,v}^{{\rm{up}},y}$ is shown as Eq. \eqref{xi_uv_up_y} at the bottom of this page.
	 Comparing $\xi_{u,v}^{{\rm{up}},x}$ and $\zeta_v$, we notice that they have many elements in common, i.e., ${{{\bar \beta }_{b,v,k}}{{\left| {{\hat{\bar h}_{b,v,k}}} \right|}^2}}$, except that the elements are preceded by different relaxation factors, so the two are very tightly correlated, while the correlation between $\xi_{u,v}^{{\rm{up}},y}$ and ${\xi _v^{{\rm{down}}}}$ is relatively weak. We use different approximation treatments for the two terms as follows
	 \setcounter{equation}{56}
	\begin{equation}\label{approx}
		\begin{aligned}
			\mathbb{E}\left\{ \frac{{\xi_{u,v}^{\text{up}}}}{{\xi _v^{{\rm{down}}}}} \right\}&=\mathbb{E}\left\{ \frac{\xi_{u,v}^{{\rm{up}},x}}{\xi _{v}^{\text{down}}} \right\}+\mathbb{E}\left\{ \frac{\xi_{u,v}^{{\rm{up}},y}}{\xi _{v}^{\text{down}}} \right\},\\
			&\approx \frac{\mathbb{E}\left\{\xi_{u,v}^{{\rm{up}},x} \right\}}{\mathbb{E}\left\{ \xi _{v}^{\text{down}} \right\}}+\mathbb{E}\left\{ \xi_{u,v}^{{\rm{up}},y}\right\}\mathbb{E}\left\{ \frac{1}{\xi _{v}^{\text{down}}} \right\}.
		\end{aligned}
	\end{equation}

	In the case of $v=u$, we have 
	\begin{equation}\label{xi_up_x_equ}
		\begin{aligned}
			& \quad\,\,\,\,{\mathbb{E}\left\{\xi_{u,v}^{{\rm{up}},x} \right\}} \\ 
			& \overset{(a)}{=}M^2
			{\mathbb{E}}\left\{ {{\hat{\bar{\bf{h}}}}^{\rm{H}}_{{\Lambda}_{u}}}
			{{\bf{B }}^{u}_{\Lambda_u}}{{\bf{B }}^{u,{\rm{H}}}_{\Lambda_u}}
			{\tilde{\bar{\bf{h}}}}_{\Lambda_u}{\tilde{\bar{\bf{h}}}}^{\rm{H}}_{\Lambda_u}
			{{\bf{B }}_{\Lambda_u}}{{\bf{B }}^{\rm{H}}_{\Lambda_u}}
			{{\hat{\bar{\bf{h}}}}_{{\Lambda}_{u}}}  \right\},\\
			& \overset{(b)}{=}M^2
			{\mathbb{E}}\left\{ {{\hat{\bar{\bf{h}}}}^{\rm{H}}_{{\Lambda}_{u}}}
			{{\bf{B }}^{u}_{\Lambda_u}}{{\bf{B }}^{u,{\rm{H}}}_{\Lambda_u}}
			{\mathbb{E}}\left\{{\tilde{\bar{\bf{h}}}}_{\Lambda_u}{\tilde{\bar{\bf{h}}}}^{\rm{H}}_{\Lambda_u}  \right\}
			{{\bf{B }}_{\Lambda_u}}{{\bf{B }}^{\rm{H}}_{\Lambda_u}}
			{{\hat{\bar{\bf{h}}}}_{{\Lambda}_{u}}}  \right\},\\
			& =  M^2
			{\mathbb{E}}\left\{ {{\hat{\bar{\bf{h}}}}^{\rm{H}}_{{\Lambda}_{u}}}
			{{\bf{B }}^{u}_{\Lambda_u}}{{\bf{B }}^{u,{\rm{H}}}_{\Lambda_u}}
			{{\tilde{\bf{R}}_{\Lambda_u}}}
			{{\bf{B }}_{\Lambda_u}}{{\bf{B }}^{\rm{H}}_{\Lambda_u}}
			{{\hat{\bar{\bf{h}}}}_{{\Lambda}_{u}}}  \right\},\\
			& =M^2 {\mathbb{E}}\left\{ {{\hat{\bar{\bf{h}}}}^{\rm{H}}_{{\Lambda}_{u}}}
			{ {\hat{\bf{R}}_{\Lambda_u}}} ^{-{\rm{1}}/{2}}
			{{\bf{S }}_{u,u}}{ {\hat{\bf{R}}_{\Lambda_u}}} ^{-{\rm{1}}/{2}}
			{{\hat{\bar{\bf{h}}}}_{{\Lambda}_{u}}}  \right\},\\
		\end{aligned}
	\end{equation}
	where $(a)$ replaces scalar accumulation with matrix multiplication, and $(b)$ follows from the independence between ${{\hat{\bar{\bf{h}}}}_{{\Lambda}_{v}}} $ and ${\tilde{\bar{\bf{h}}}}_{\Lambda_u}$.

	For ${\mathbb{E}\left\{\xi_{u,v}^{{\rm{up}},x} \right\}}$ with the case of $v \ne u$, we have
	\begin{equation}\label{xi_up_x}
	\begin{aligned}
		& \quad\,\,\,\,{\mathbb{E}\left\{\xi_{u,v}^{{\rm{up}},x} \right\}} \\ & =
			M^2{\mathbb{E}}\left\{ {{\hat{\bar{\bf{h}}}}^{\rm{H}}_{{\Lambda}_{v}}}
			{{\bf{B }}^{u}_{\Lambda_v}}{{\bf{B }}^{u,{\rm{H}}}_{\Lambda_v}}
			{\bf{A}}_{v}{\bar{\bf{h}}}_{u}{\bar{\bf{h}}}^{\rm{H}}_{u}{\bf{A}}^{\rm{T}}_{v}
			{{\bf{B }}_{\Lambda_v}}{{\bf{B }}^{\rm{H}}_{\Lambda_v}}
			{{\hat{\bar{\bf{h}}}}_{{\Lambda}_{v}}}  \right\},\\
			&\overset{(a)}{=}M^2
			{\mathbb{E}}\left\{ {{\hat{\bar{\bf{h}}}}^{\rm{H}}_{{\Lambda}_{v}}}
			{{\bf{B }}^{u}_{\Lambda_v}}{{\bf{B }}^{u,{\rm{H}}}_{\Lambda_v}}
			{\bf{A}}_{v}  {\mathbb{E}}\left\{{\bar{\bf{h}}}_{u}{\bar{\bf{h}}}^{\rm{H}}_{u}
			\right\}
			{\bf{A}}^{\rm{T}}_{v}
			{{\bf{B }}_{\Lambda_v}}{{\bf{B }}^{\rm{H}}_{\Lambda_v}}
			{{\hat{\bar{\bf{h}}}}_{{\Lambda}_{v}}}  \right\},\\
			&\overset{(b)}{=} M^2
			{\mathbb{E}}\left\{ {{\hat{\bar{\bf{h}}}}^{\rm{H}}_{{\Lambda}_{v}}}
			{{\bf{B }}^{u}_{\Lambda_v}}{{\bf{B }}^{u,{\rm{H}}}_{\Lambda_v}}
			{{\bf{B }}_{\Lambda_v}}{{\bf{B }}^{\rm{H}}_{\Lambda_v}}
			{{\hat{\bar{\bf{h}}}}_{{\Lambda}_{v}}}  \right\},\\
			 &=  M^2	{\mathbb{E}}\left\{ {{\hat{\bar{\bf{h}}}}^{\rm{H}}_{{\Lambda}_{v}}}
			{ {\hat{\bf{R}}_{\Lambda_v}}} ^{-{\rm{1}}/{2}}
			{{\bf{S }}_{u,v}}{ {\hat{\bf{R}}_{\Lambda_v}}} ^{-{\rm{1}}/{2}}
			{{\hat{\bar{\bf{h}}}}_{{\Lambda}_{v}}}  \right\},
	\end{aligned}
	\end{equation}
where $(a)$ follows from the independence between ${{\hat{\bar{\bf{h}}}}_{{\Lambda}_{v}}} $ and ${\bar{\bf{h}}}_{u}$, $\forall v\ne u$.
	Recall that the diagonal elements of ${\bf{R}}_{u} = {\mathbb{E}}\left\{{\bar{\bf{h}}}_{u}{\bar{\bf{h}}}^{\rm{H}}_{u}
	\right\}$ are all $1$, indicating that the port auto-correlation coefficient is $1$, i.e., $\rho_{u,b,b}^{l,l}  = 1$. Thus, ${\bf{A}}_{v}{\bf{R}}_{u}{\bf{A}}_{v}^{\rm{T}} = {\bf{I}}_{K_v}$, and $(b)$ holds.
	
	Define $\check{\check{\mathbf{ h}} } _{{\Lambda}_{v}} =  { {\hat{\bf{R}}_{\Lambda_v}}} ^{-{\rm{1}}/{2}}
	{{\hat{\bar{\bf{h}}}}_{{\Lambda}_{v}}}$. We have $\check{\check{\mathbf{ h}} } _{{\Lambda}_{v}} \sim \mathbb{CN}({\bf{0}},{\bf{I}}_{K_v})$. Therefore, from Eq. \eqref{xi_up_x_equ} and Eq. \eqref{xi_up_x},
	\begin{equation}
	{\mathbb{E}\left\{\xi_{u,v}^{{\rm{up}},x} \right\}} = M^2{\mathbb{E}\left\{ \check{\check{\mathbf{ h}} }^{\rm{H}} _{{\Lambda}_{v}} {\bf{S}}_{u,v} \check{\check{\mathbf{ h}} } _{{\Lambda}_{v}} \right\}} = M^2{{\rm{tr}}\left( {\mathbf{S}_{u,v}} \right)}.
	\end{equation}
	
	Similar to Eq. \eqref{zetau_appendix}, ${{\xi _v^{{\rm{down}}}}}$ can be re-expressed as ${{\xi _v^{{\rm{down}}}}} = M^2 \left({\check{\bf{h}}^{\rm{H}}_{{\Lambda}_{v}}}{{\boldsymbol{\Sigma }}_{v,\text{s}}}{\check{\bf{h}}_{{\Lambda}_{v}}}\right)^2$,
%	\begin{equation}\label{down}
%		{{\xi _v^{{\rm{down}}}}} =  \left({\check{\bf{h}}^{\rm{H}}_{{\Lambda}_{v}}}{{\boldsymbol{\Sigma }}_{v,\text{s}}}{\check{\bf{h}}_{{\Lambda}_{v}}}\right)^2,
%	\end{equation}
	where $ {\check{\bf{h}}_{{\Lambda}_{v}}} \sim \mathbb{CN}{\left( {{\bf{0}},{{\bf{I}}_{K_v}}} \right)}$, and  ${{\boldsymbol{\Sigma }}_{v,\text{s}}}$ is the eigenvalue matrix of  ${{\bf{S}}_{v}}$ defined in Eq. \eqref{S_u}. Therefore, 
	\begin{equation}
		\begin{aligned}
			{\mathbb{E}\left\{ \xi _{v}^{\text{down}} \right\}}&=M^2{\mathbb{E}\left\{ {{\left| \sum\limits_{i=1}^{K_v}{{{\lambda }_{v,i}}{{\left| \check{h}_{\Lambda_v,i} \right|}^{2}}} \right|}^{2}} \right\}},\\
			&={2M^2\sum\limits_{p=1}^{K_v}{\lambda _{v,p}^{2}}+M^2\sum\limits_{p=1}^{K_v}{\sum\limits_{q\ne p}^{K_v}{{{\lambda }_{v,p}}{{\lambda }_{v,q}}}}},\\
			&={M^2{{\left| {\rm{tr}}\left( {\mathbf{S}_{v}} \right) \right|}^{2}}+M^2\left\| {{\mathbf{S }}_{v}} \right\|_{\rm{F}}^{2}},\\
		\end{aligned}
	\end{equation}
and
	\begin{equation}\label{first_appro}
		\begin{aligned}
			\frac{\mathbb{E}\left\{ \xi_{u,v}^{{\rm{up}},x} \right\}}{\mathbb{E}\left\{ \xi _{v}^{\text{down}} \right\}}
			=\frac{{\rm{tr}}\left( \mathbf{S}_{u,v} \right)}{{{\left| {\rm{tr}}\left( \mathbf{S}_{v} \right) \right|}^{2}}+\left\| {{\mathbf{S }}_{v}} \right\|_{\rm{F}}^{2}}.
		\end{aligned}
	\end{equation}

	\newcounter{TempEqCnt7}                         % 创建临时变量TempEqCnt
	\setcounter{TempEqCnt7}{\value{equation}} % 将当前公式序号 赋给TempEqCnt
	\setcounter{equation}{62}  % 当前公式序号变为x，x等于长公式应有的序号减1.
	\begin{figure*}[ht] % hb底部，ht为头部
		\centering % 公式居中
		%\hrulefill % 添加一条水平线
		\vspace*{-20pt} % 调整线与公式之间的距离
		\textcolor{black}{
			\begin{equation} \label{xi_uv_up_y_expectation_1}
				\begin{aligned}
					\mathbb{E }\left\{\xi_{u,v}^{{\rm{up}},y}\right\} &= 
					 M^2\sum\limits_{b' = 1}^B {\sum\limits_{b \ne b'}^B {\sum\limits_{l \in {\Lambda _{b,u}}} {\sum\limits_{l' \in {\Lambda _{b',u}}} {
									{{{\bar \beta }_{b,u,l}}{{\bar \beta }_{b',u,l'}}
										\mathbb{E}\left \{ 
										{{\tilde {\bar h}}_{b,u,l}}
										\tilde {\bar h}_{b',u,l'}^*
										\right \} 
										\mathbb{E}\left \{ 
										{{\hat {\bar h}}_{b,u,l}}
										\hat {\bar h}_{b',u,l'}^*
										\right \} 
					} } } } } , \\
					&= M^2	\sum\limits_{b' = 1}^B {\sum\limits_{b \ne b'}^B {\sum\limits_{l \in {\Lambda _{b,u}}} {\sum\limits_{l' \in {\Lambda _{b',u}}} {
									{{{\bar \beta }_{b,u,l}}{{\bar \beta }_{b',u,l'}}
									\varepsilon^2\left(1- \varepsilon ^2\right){\rho}_{u,b,b'}^{l,l',2}
										 ,\quad {\text{if }}\, v = u
					} } } } } .
				\end{aligned}
			\end{equation}
	\begin{equation} \label{xi_uv_up_y_expectation_2}
		\begin{aligned}
			\mathbb{E }\left\{\xi_{u,v}^{{\rm{up}},y}\right\} &= M^2
			 \sum\limits_{b' = 1}^B {\sum\limits_{b  \ne  b'}^B {\sum\limits_{l \in {\Lambda _{b,v}}} {\sum\limits_{l' \in {\Lambda _{b',v}}} {\sqrt {{{\bar \beta }_{b,v,l}}{{\bar \beta }_{b,u,l}}{{\bar \beta }_{b',v,l'}}{{\bar \beta }_{b',u,l'}}}
							\mathbb{E}\left \{  
							{{\bar {h}}}_{b,u,l}^* 
							{{{\bar{ h}}}_{b',u,l'}} 
							\right \} 
							\mathbb{E}\left \{  
							{\hat{\bar {h}}}_{b,v,l}^* 
							{{\hat{\bar{ h}}}_{b',v,l'}} 
							\right \} 	}  } } }  ,\\			
			& = M^2\sum\limits_{b' = 1}^B {\sum\limits_{b  \ne  b'}^B {\sum\limits_{l \in {\Lambda _{b,v}}} {\sum\limits_{l' \in {\Lambda _{b',v}}} {\sqrt {{{\bar \beta }_{b,v,l}}{{\bar \beta }_{b,u,l}}{{\bar \beta }_{b',v,l'}}{{\bar \beta }_{b',u,l'}}}
							\left(1- \varepsilon ^2\right){\rho}_{u,b,b'}^{l,l'}
							{\rho}_{v,b,b'}^{l,l'}  , \quad {\text{if }}\, v \ne u
			}  } } }   .
		\end{aligned}
	\end{equation}  }
		%\vspace*{8pt} % 调整线与公式之间的距离
		\hrulefill % 添加一条水平线
	\end{figure*}
	\setcounter{equation}{\value{TempEqCnt7}} % 把TempEqCnt中存的公式序号赋回给当前公式序号

	Next, we compute $\mathbb{E}\left\{ \xi_{u,v}^{{\rm{up}},y} \right\} $  with the case of $v=u$ and $v \ne u$ as shown in Eq. \eqref{xi_uv_up_y_expectation_1} and Eq. \eqref{xi_uv_up_y_expectation_2} at the top of this page, respectively. 
	Thus, we obtain $\mathbb{E}\left\{ {\xi_{u,v}^{{\rm{up}},y}} \right\} = {{\delta }_{u,v}} $ defined in Eq. \eqref{delta}.
	
	For $\mathbb{E}\left\{ \frac{1}{\xi _{v}^{\text{down}}} \right\}$,
	 we can obtain for $\rho_v >2$ from Eq. \eqref{PDF_zeta_u},
	 \setcounter{equation}{64}
	\begin{equation}
		\label{expectation2}
		\begin{aligned}
			&\quad\,\mathbb{E}\left\{ \frac{1}{\xi _{v}^{\text{down}}} \right\} = \frac{1}{M^2}\mathbb{E}\left\{ {\frac{1}{{\zeta _v^2}}} \right\} ,\\
			&= \frac{1}{M^2}\sum\limits_{k = 0}^\infty  {\frac{{{\alpha_{v,k}}}}{{\Gamma \left( {\rho_v  + k} \right){{\left( {2\beta_v } \right)}^{\rho_v  + k}}}}} \int_0^\infty  {{x^{\rho_v  + k - 3}}} {e^{ - \frac{1}{{2\beta_v }}x}}\mathrm{d}x ,\\ 
			&= \frac{1}{M^2}\!\!\sum\limits_{k = 0}^\infty \! {\frac{{{\alpha_{v,k}}}}{{\Gamma \left( {\rho_v  + k} \right){{\left( {2\beta_v } \right)}^{\rho_v  + k}}}}} {\left( {2\beta_v } \right)^{\rho_v  + k - 2}}\Gamma \left( {\rho_v  + k - 2} \right), \\
		\end{aligned}
	\end{equation}
	which simplifies to ${{\eta }_{v}}$ in Eq. \eqref{eta}.
	Up to this point, the approximate expression in Eq. \eqref{lemma2_appro} can be obtained.

	\ifCLASSOPTIONcaptionsoff
	\newpage
	\fi


	
	%\bibliographystyle{IEEEtranTCOM}
	
%	\bibliographystyle{IEEEtran}
%	%\bibliographystyle{IEEEtranTCOM}
%	\bibliography{ref}

	%\bibliographystyle{IEEEtran}
	%%%\bibliographystyle{IEEEtranTCOM}
	%\bibliography{References}

\begin{IEEEbiography}[{\includegraphics[width=1in,height=1.25in,clip,keepaspectratio]{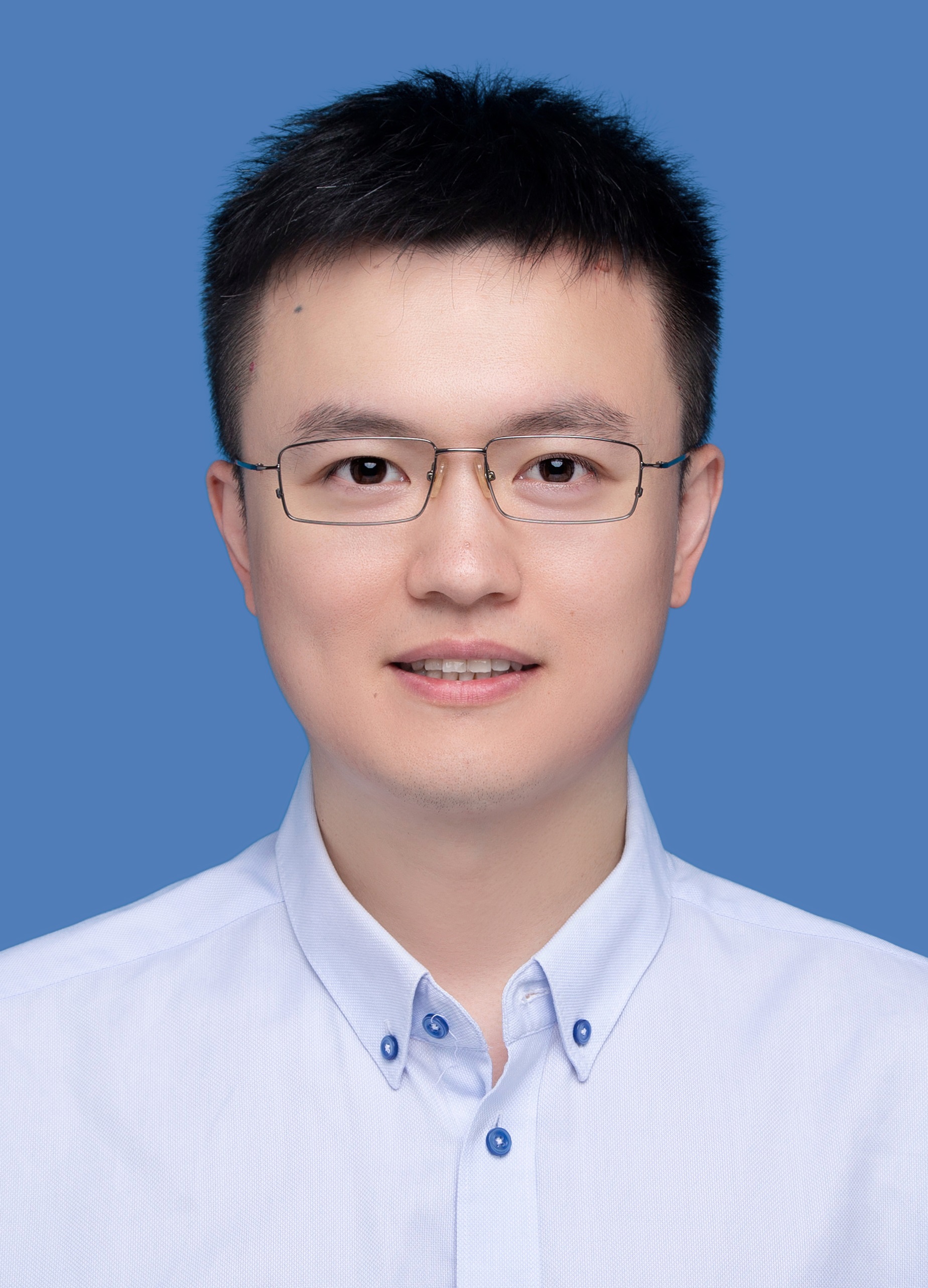}}]{Cheng Zhang}
	(Member, IEEE) received the B.Eng. degree from Sichuan University, Chengdu, China, in June 2009, the M.Sc. degree from the Xi’an Electronic Engineering Research Institute (EERI), Xi’an, China, in May 2012, and the Ph.D. degree from Southeast University (SEU), Nanjing, China, in Dec. 2018. From Nov. 2016 to Nov. 2017, he was a Visiting Student with the University of Alberta, Edmonton, AB, Canada.
	
	From June 2012 to Aug. 2013, he was a Radar Signal Processing Engineer with Xi’an EERI. Since Dec. 2018, he has been with SEU, where he is currently an Associate Professor, and supported by the Zhishan Young Scholar Program of SEU. His current research interests include space-time signal processing and machine learning-aided optimization for B5G/6G wireless communications. He has authored or co-authored more than 50 IEEE journal papers and conference papers. He was the recipient of the excellent Doctoral Dissertation of the China Education Society of Electronics in Dec. 2019, that of Jiangsu Province in Dec. 2020, and the Best Paper Awards of 2023 IEEE WCNC and 2023 IEEE WCSP.
\end{IEEEbiography}

\begin{IEEEbiography}[{\includegraphics[width=1in,height=1.25in,clip,keepaspectratio]{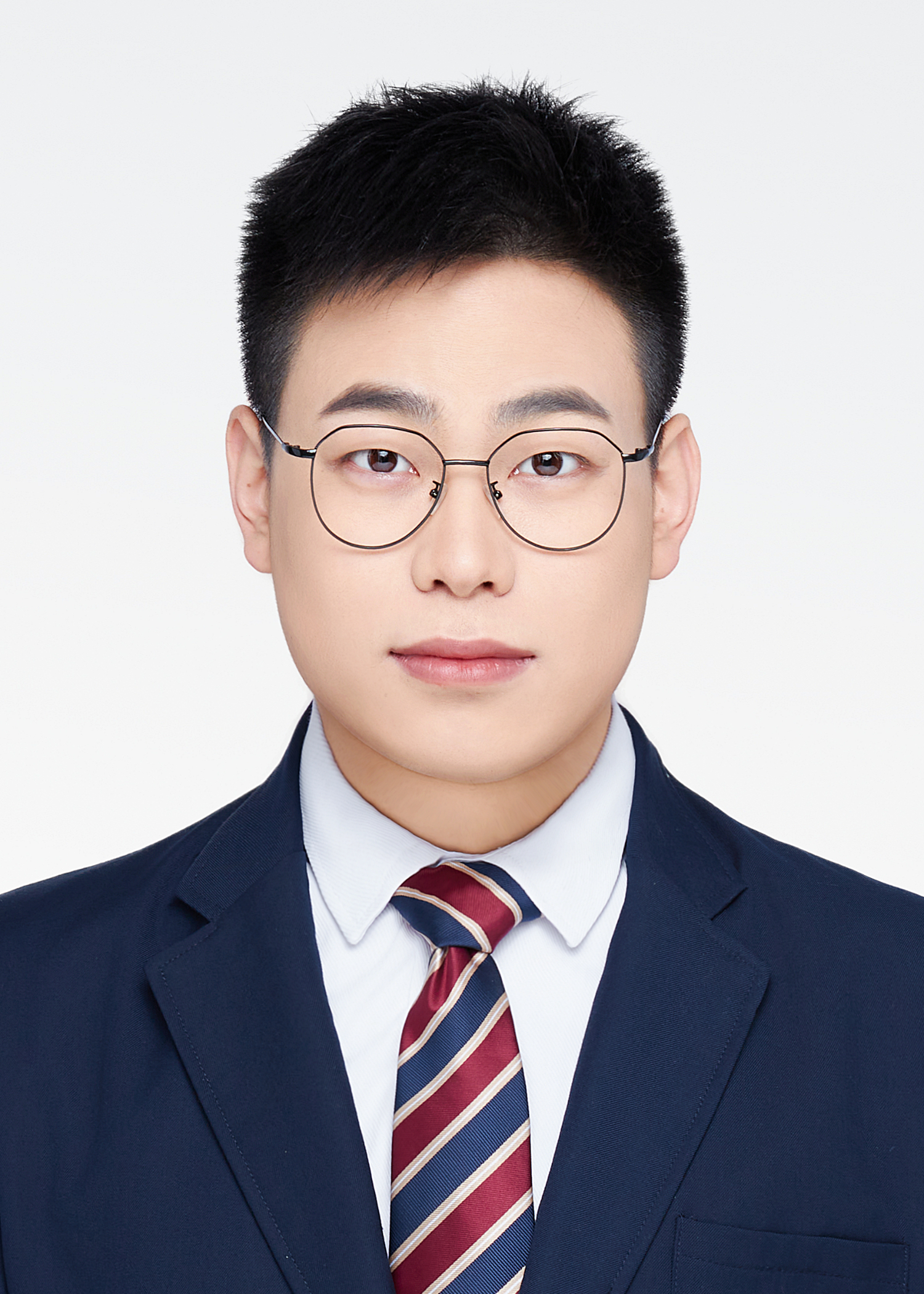}}]{Pengguang Du}
		(Graduate Student Member, IEEE) received the B.Eng. degree in communication engineering from Jilin University, Changchun, China, in 2021. And he is currently pursuing the Ph.D. degree in information and communication engineering with the School of Information Science and Engineering, Southeast University, Nanjing, China. His research interests mainly focus on massive MIMO channel acquisition and intelligent wireless communications.
\end{IEEEbiography}

\begin{IEEEbiography}[{\includegraphics[width=1in,height=1.25in,clip,keepaspectratio]{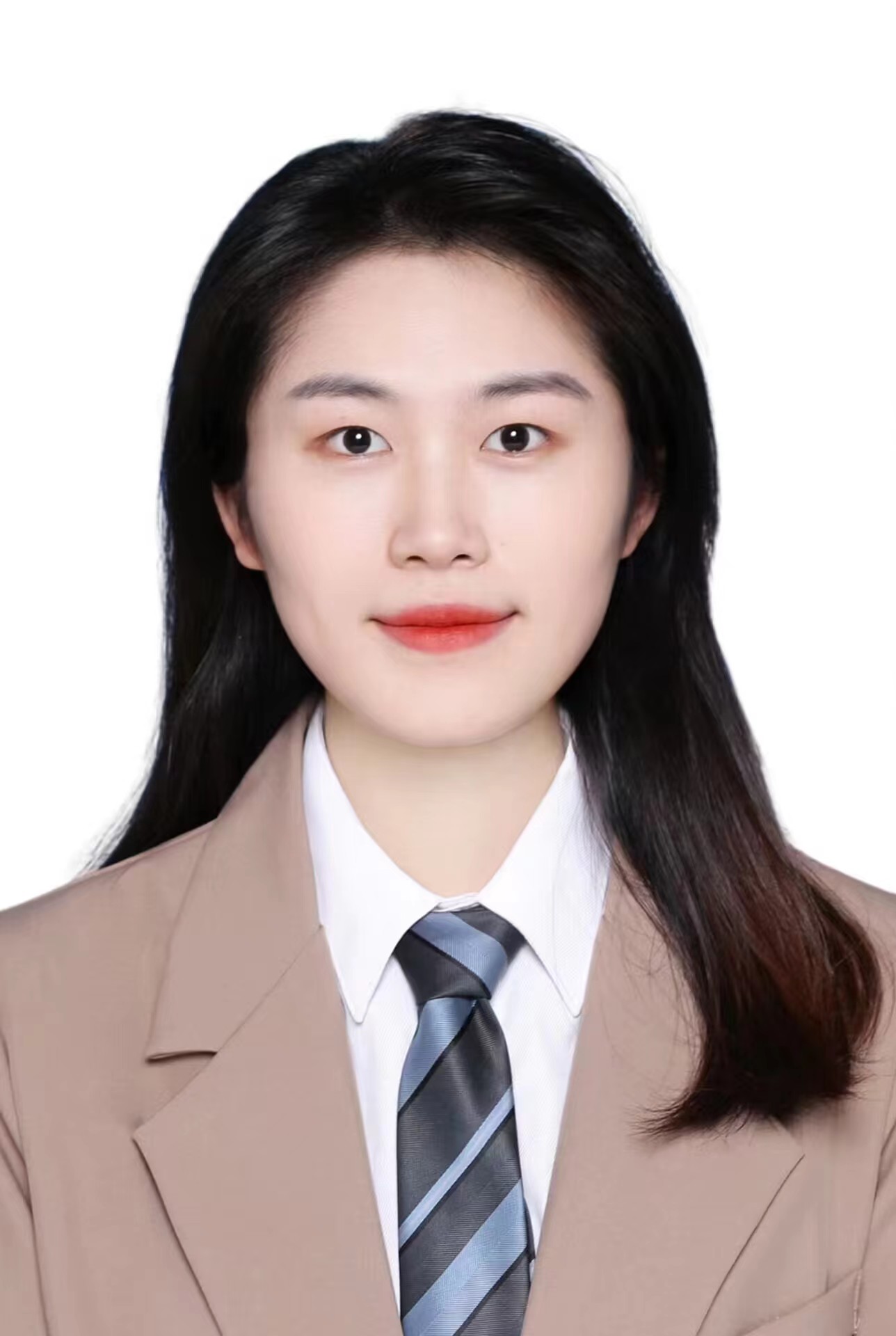}}]{Minjie Ding}
		received the B.Eng. degree in information engineering from the School of Electronic Information and Engineering, Nanjing University of Aeronautics and Astronautics, Nanjing, China, in 2020, where she is currently pursuing the M.Sc. degree in information and communication engineering with the School of Information Science and Engineering, Southeast University. Her research interests mainly focus on low overhead massive MIMO channel acquisition.
\end{IEEEbiography}

\begin{IEEEbiography}[{\includegraphics[width=1in,height=1.25in,clip,keepaspectratio]{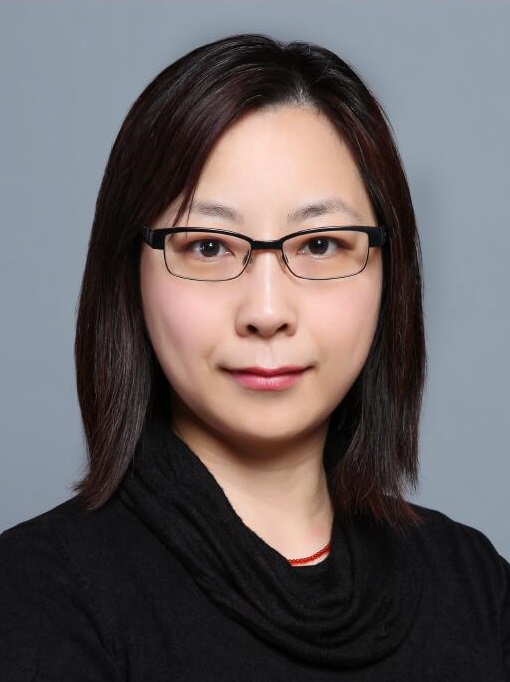}}]{Yindi Jing}
	(Senior Member, IEEE) received the B.Eng. and M.Eng. degrees in automatic control
	from the University of Science and Technology of China, Hefei, China, in 1996 and 1999, respectively. She received the M.Sc. degree and the Ph.D. in electrical engineering from California Institute of Technology, Pasadena, CA, in 2000 and 2004, respectively. From October 2004 to August 2005, she was a postdoctoral scholar at the Department of Electrical Engineering of California Institute of Technology. From February 2006 to June 2008, she was a postdoctoral scholar at the Department of Electrical Engineering and Computer Science of the University of California, Irvine. In 2008, she joined the Electrical and Computer Engineering Department of the University of Alberta, where she is currently a professor. She was an Associate Editor for IEEE TRANSACTIONS ON WIRELESS COMMUNICATIONS and a Senior Area Editor for IEEE SIGNAL PROCESSING LETTERS. She was a member of the IEEE Signal Processing Society Signal Processing for Communications and Networking (SPCOM) Technical Committee and a member of the NSERC Discovery Grant Evaluation Group for Electrical and Computer Engineering. Her research interests are in wireless communications and signal processing.
\end{IEEEbiography}

\begin{IEEEbiography}[{\includegraphics[width=1in,height=1.25in,clip,keepaspectratio]{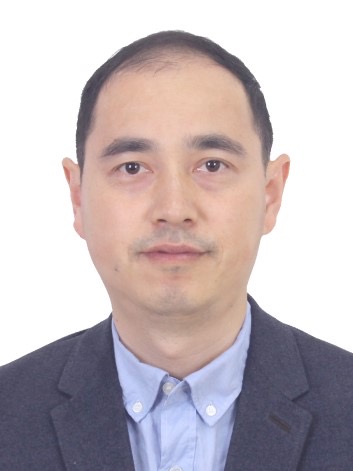}}]{Yongming Huang}
	(M’10-SM’16) received the B.S. and M.S. degrees from Nanjing University, Nanjing, China, in 2000 and 2003, respectively, and the Ph.D. degree in electrical engineering from Southeast University, Nanjing, in 2007.
	
	Since March 2007 he has been a faculty in the School of Information Science and Engineering, Southeast University, China, where he is currently a full professor. He has also been the Director of the Pervasive Communication Research Center, Purple Mountain Laboratories, since 2019. From 2008 to 2009, he was visiting the Signal Processing Lab, Royal Institute of Technology (KTH), Stockholm, Sweden. He has published over 200 peer-reviewed papers, hold over 80 invention patents. His current research interests include intelligent 5G/6G mobile communications and millimeter wave wireless communications. He submitted around 20 technical contributions to IEEE standards, and was awarded a certiﬁcate of appreciation for outstanding contribution to the development of IEEE standard 802.11aj. He served as an Associate Editor for the IEEE Transactions on Signal Processing and a Guest Editor for the IEEE Journal on Selected Areas in Communications. He is currently an Editor-at-Large for the IEEE Open Journal of the Communications Society.
\end{IEEEbiography}


\end{document}